\documentclass[aps,pra,10pt,onecolumn]{revtex4-2}
\usepackage{amsfonts,amsmath,amssymb,amsthm}
\usepackage{braket}
\usepackage{dcolumn,longtable}
\usepackage{graphicx}
\usepackage{xspace}
\usepackage{tikz}

\newcommand{\QSL}{QSL\xspace}
\newcommand{\ML}{ML\xspace}
\newcommand{\MT}{MT\xspace}
\newcommand{\Chau}{Chau\xspace}
\newcommand{\LC}{LC\xspace}
\newcommand{\LZ}{LZ\xspace}
\newcommand{\CZ}{CZ\xspace}
\newcommand{\dualML}{dual {\ML}\xspace}

\DeclareMathOperator{\cf}{{\mathbf 1}}
\DeclareMathOperator{\Oh}{O}
\DeclareMathOperator{\sgn}{sgn}

\newcommand{\eplus}[2]{\braket{[{#1}^{+}]^{#2}}}
\newcommand{\opt}{{\rm opt}}
\newcommand{\crit}{{\rm crit}}
\newcommand{\llangle}{\langle\!\langle}
\newcommand{\rrangle}{\rangle\!\rangle}

\newtheorem{lemma}{Lemma}
\newtheorem{theorem}{Theorem}
\newtheorem{corollary}{Corollary}
\newtheorem{remark}{Remark}

\begin{document}
\title{A Unifying Quantum Speed Limit For Time-Independent Hamiltonian
 Evolution}
\author{H. F. Chau}
\email[Corresponding author, email: ]{{\tt hfchau@hku.hk}}
\author{Wenxin Zeng}
\affiliation{Department of Physics, University of Hong Kong, Pokfulam Road,
 Hong Kong}
\date{\today}

\begin{abstract}
 Quantum speed limit (\QSL) is the study of fundamental limits on the
 evolution time of quantum systems.  For instance, under the action of a
 time-independent Hamiltonian, the evolution time between an initial and a
 final quantum state obeys various mutually complementary lower bounds.  They
 include the Mandelstam-Tamm bound, the Margolus-Levitin bound, the Luo-Zhang
 bound, the Lee-Chau bound together with the \dualML bound introduced by Ness
 and coworkers.  Here we show that the Mandelstam-Tamm bound can be obtained
 by optimizing the Lee-Chau bound over a certain parameter.  More importantly,
 we report a \QSL that includes all the above bounds as special cases before
 optimizing over the physically meaningless reference energy level of a
 quantum system.  This unifying bound depends on a certain parameter $p$.  For
 any fixed $p$, we find all pairs of time-independent Hamiltonian and initial
 pure quantum state that saturate this unifying bound.  More importantly,
 these pairs allow us to compute this bound accurately and efficiently using
 an oracle that returns certain $p$th moments related to the absolute value of
 energy of the quantum state.  Moreover, this oracle can be simulated by a
 computationally efficient and accurate algorithm for finite-dimensional
 quantum systems as well as for certain infinite-dimensional quantum states
 with bounded and continuous energy spectra.  This makes our computational
 method feasible in a lot of practical situations.  We further compare the
 performance of this bound for the case of a fixed $p$ as well as the case of
 optimizing over $p$ with existing \QSL{s}.  We find that if the dimension of
 the underlying Hilbert space is $\lesssim 2000$, our unifying bound optimized
 over $p$ can be computed accurately in a few minutes using Mathematica code
 with just-in-time compilation in a typical desktop.  Besides, this optimized
 unifying \QSL is at least as good as all the existing ones combined and can
 occasionally be a few percent to a few times better.
\end{abstract}

\maketitle

\section{Introduction}
\label{Sec:Intro}
 Quantum speed limit (\QSL) is the study of fundamental limits in quantum
 information processing speed~\cite{Lloyd00}.  The first and probably the
 most well-known \QSL is the Mandelstam-Tamm (\MT) bound.  It says that the
 evolution time $\tau$ under the action of a time-independent Hamiltonian
 obeys the inequality
\begin{equation}
 \frac{\tau}{\hbar} \ge \frac{\cos^{-1}(\sqrt{\epsilon})}{\Delta E} ,
 \label{E:MT_bound}
\end{equation}
 where $\epsilon$ is the fidelity between the initial and final quantum state
 and $\Delta E$ is the energy standard deviation of the quantum
 state~\cite{MT45}.  Actually, the \MT bound was discovered well before the
 quantum information era and the coinage of the term \QSL.  Moreover, it can
 be extended to cover the cases of evolution under a time-dependent
 Hamiltonian or open system dynamics~\cite{BCM96}.

 A lot of \QSL bounds have been discovered.  A few are applicable to
 time-dependent Hamiltonians as well as open systems.  Relations with quantum
 control, entanglement, resource theory as well as the so-called time-fractional Schr\"{o}dinger equation have also been
 explored~\cite{[{See, for example, }]Frey16,*DC17,T13,C13,MP16,C22,JLB23,PMSB23,PSMDP23,WLLFQW23}.
 Recently, a collection of articles on \QSL and its applications was published
 in a special section of a journal~\cite{NJPissue}.  In there,
 Takahashi considered not just lower bound but also upper bound on evolution
 time of quantum system~\cite{T22}; also Aifer and Deffner relates \QSL to
 energy efficient implementation of quantum gates~\cite{AD22}.
 Along a different line, Shanahan \emph{et al.}~\cite{SCMC18} as well as
 Okuyama and Ohzeki~\cite{OO18} found that there is a classical correspondence
 to certain \QSL and concluded that \QSL is a universal dynamical property of
 Hilbert space rather than a pure quantum phenomenon.
 In this paper, we go back to the basics by studying \QSL{s} of
 time-independent Hamiltonians for closed systems.  We make the following four
 major contributions in this study.

 First, we report a \QSL that we called the \CZ bound for easy reference.
 This bound generalizes a number of existing \QSL{s} for time-independent
 Hamiltonian evolution.  They include the \MT bound~\cite{MT45}, the
 Margolus-Levitin (\ML) bound~\cite{PHYSCOMP96,ML98,GLM03}, the Luo-Zhang
 (\LZ) bound~\cite{LZ05}, the Lee-Chau (\LC) bound~\cite{LC13} as well as the
 \dualML bound introduced by Ness \emph{et al.}~\cite{NAS22}.  Actually, the
 \LZ, \LC and \CZ bounds are families of bounds each depending on a parameter
 $p$~\cite{LZ05,LC13}.  In addition, for a given $p$, the \LC bound is
 optimized over the physically meaningless reference energy level $E_r$ of the
 quantum state~\cite{LC13}.  And for the \CZ bound, we will see that it is
 optimized over both $E_r$ and another variable named $\theta$.  In this
 paper, we prove that the \CZ bound generalizes the \MT, \ML, \dualML, \LZ and
 \LC bounds in two steps.  Because all but the \MT bound have similar forms,
 our first step is to show that the \CZ bound can be reduced to all but the
 \MT bound above by fixing one or more of these three parameters to certain
 specific values instead of optimizing over them.  Surprisingly, even though
 the forms of the \ML, \dualML, \LZ, \LC and \CZ bounds are very different
 from that of the \MT bound, in our second step, we prove that the \MT bound
 is in fact a result of the \LC bound by optimizing over the parameter $p$.

 Our second major contribution is to study the necessary and sufficient
 conditions for a pair of time-independent Hamiltonian and initial pure
 quantum state to saturate the \CZ bound, and as a result explicitly write
 them down in a computationally usable form.  This analysis can be extended to
 give the set of all time-independent Hamiltonian and initial pure state pairs
 that saturate the \LC bound as well.

 Our third contribution is a numerically efficient (that is, polynomial-time
 computable) and accurate (that is, numerically stable and rounding error is
 not serious) method to compute the \CZ bound (and hence also the \LC bound)
 under the assumption that there is an efficient and accurate method to
 evaluate the minimum $p$th moment of the absolute value of energy of the
 quantum state and the corresponding $p$th signed moment of the absolute
 value of energy of the state.  This assumption is true for finite-dimensional
 quantum systems as well as certain infinite-dimensional systems with bounded
 and continuous energy spectra.  Hence, the \CZ bound is computationally
 feasible in almost all realistic situations.  Interestingly, this method is
 closely related to the Hamiltonian and quantum state pairs that saturate the
 \CZ bound.

 Our last contribution is an extensive comparison of the evolution time lower
 bound and actual runtime between the \CZ bound that is further optimized over
 the parameter $p$ and other existing bounds.  We discover that even when the
 dimension of the underlying Hilbert space is as large as $2048$, this
 optimized \CZ bound can be computed in a few minutes using Mathematica with
 just-in-time compilation installed in a typical desktop computer.
 Furthermore, the bound obtained can be a few percent to a few times better
 than the best existing one.  In this regard, the \CZ bound is the best choice
 in practice.

 In what follows, we first state various \QSL{s} that we are going to extend
 in Sec.~\ref{Sec:prior_art}.  We then report several auxiliary results in
 Sec.~\ref{Sec:auxiliary}.  These results allow us to prove the \CZ bound and
 to study the necessary and sufficient conditions for its saturation in
 Sec.~\ref{Sec:bound}.  In doing so, we find the set of all time-independent
 Hamiltonian and initial pure state pairs that saturate the \CZ and the \LC
 bounds, respectively.  Also in Sec.~\ref{Sec:bound}, we show that by
 optimizing over the parameter $p$, the \LC bound is reduced to the \MT bound.
 Besides, the \CZ bound extends the \ML, \dualML, \LZ and \LC bounds.  Hence,
 the \CZ bound unifies the \MT, \ML, \dualML, \LC and \LC bounds.  Next, we
 report an accurate and efficient method to numerically compute the \CZ as
 well as the \LC bounds given an oracle returning the minimum $p$th moment of
 the absolute value of energy of the quantum state and the corresponding $p$th
 signed moment of the absolute value of energy of the state in
 Sec.~\ref{Sec:numeric}.  We also show that this efficient oracle exists in
 the sense that it can be replaced by an efficient and accurate algorithm for
 finite-dimensional as well as certain infinite-dimensional systems.  As a
 byproduct, we report there a simple expression for the \CZ bound for
 two-dimensional quantum systems.  Using the numerical method developed in
 Sec.~\ref{Sec:numeric}, we compare the performance of the \CZ bound optimized
 over the parameter $p$ with existing ones in practice in
 Sec.~\ref{Sec:performance}.  Finally, we summarize our findings in
 Sec.~\ref{Sec:conclusion}.

\section{Prior Art}
\label{Sec:prior_art}
 Here we list some of the most important \QSL bounds for quantum state
 evolution under time-independent Hamiltonian discovered so far.
 Collectively, they are the most powerful \QSL{s} for time-independent
 Hamiltonian evolution among those based on a single parameter describing the
 energy of the initial quantum state.

\begin{itemize}
 \item The \MT bound~\cite{MT45} is given by Inequality~\eqref{E:MT_bound}.
 \item Several equivalent forms of the \ML bound have been
  reported~\cite{PHYSCOMP96,ML98,GLM03,HS23,Chau23}.  The one we use here
  is~\cite{Chau23}
  \begin{equation}
   \frac{\tau}{\hbar} \ge \max_{\theta\in [-\cos^{-1}(\sqrt{\epsilon}),0]}
   \frac{\cos\theta - \sqrt{\epsilon}}{\braket{E-E_0} \sin\varphi(\theta)} ,
   \label{E:ML_bound}
  \end{equation}
  where $\braket{E-E_0}$ is the expected energy of the state relative to the
  ground state energy of the Hamiltonian, and $\varphi(\theta)$ is the unique
  root of the equation
  \begin{equation}
   \cos \varphi(\theta) - \cos\theta + [\varphi(\theta)-\theta]
   \sin\varphi(\theta) = 0
   \label{E:f_def}
  \end{equation}
  in the interval $[\max(\pi/2,|\theta|),\pi]$.
 \item The \dualML bound is the \ML-like bound reported by Ness
  \emph{et al.}~\cite{NAS22}.  Using the notation in Eq.~\eqref{E:ML_bound},
  it says that for states with bounded energy spectrum,
  \begin{equation}
   \frac{\tau}{\hbar} \ge \max_{\theta\in [-\cos^{-1}(\sqrt{\epsilon}),0]}
   \frac{\cos\theta - \sqrt{\epsilon}}{\braket{E_{\max}-E}
   \sin\varphi(\theta)}
   \label{E:dual_ML_bound}
  \end{equation}
  where $\braket{E - E_{\max}} = -\braket{E_{\max} - E}$ is the expectation
  energy of the state relative to the maximum eigenenergy of the system.
  Actually, the \dualML bound can be derived from the \ML bound through the
  following duality.  Any initial state $\ket{\Psi(0)}$ can be written in the
  form $\sum_j a_j \ket{E_j}$ with $\ket{E_j}$ being an energy eigenstate with
  energy $E_j$.  Denote its ``energy-reversed'' state $\sum_j a_j \ket{-E_j}$
  by $\ket{\tilde{\Psi}(0)}$.  Clearly, the time-evolved state $\ket{\Psi(t)}$
  equals the reversed-time evolved state $\ket{\tilde{\Psi}(-t)}$.  More
  importantly, the fidelity between $\ket{\Psi(0)}$ and $\ket{\Psi(t)}$ equals
  the fidelity between $\ket{\tilde{\Psi}(0)}$ and $\ket{\tilde{\Psi}(-t)}$,
  which in turn equals the fidelity between $\ket{\tilde{\Psi}(0)}$ and
  $\ket{\tilde{\Psi}(t)}$.  Consequently, the \ML bound for $\ket{\Psi(0)}$
  induces a bound for $\ket{\tilde{\Psi}(0)}$.  And this induced bound is the
  \dualML bound.  Thus, the \dualML bound holds for a slightly more general
  case when the energy spectrum is bounded from above.  Surely, the arguments
  reported here is general and can be used to obtain the corresponding dual
  \QSL from any given \QSL involving fidelity between the initial and final
  states.
 \item The \LZ bound~\cite{LZ05} refers to the family of \QSL{s} in the form
  \begin{equation}
   \frac{\tau}{\hbar} \ge \pi \left[ \frac{1 - \sqrt{\epsilon \left( 1 +
   4p^2/\pi^2 \right)}}{2\braket{(E - E_0)^p}} \right]^\frac{1}{p}
   \label{E:LZ_bound}
  \end{equation}
  for $0 \le p \le 2$ and $0 \le \epsilon \le \pi^2/(\pi^2 + 4p^2)$.  Surely,
  one may consider optimizing \LZ bound by taking supremum of the R.H.S. of
  Inequality~\eqref{E:LZ_bound} over $p$.  We call this the optimized \LZ
  bound.
 \item The \LC bound~\cite{LC13} is the family of inequalities
  \begin{equation}
   \frac{\tau}{\hbar} \ge \max_{E_r} \left( \frac{1 - \sqrt{\epsilon}}{
    A_p \braket{|E - E_r|^p}} \right)^\frac{1}{p} ,
   \label{E:LC_bound}
  \end{equation}
  for $0 < p \le 2$, where $A_p = \sup \{ (1-\cos x)/x^p \colon x > 0 \}$.
  Just like the optimized \LZ bound, we refer to the \LC bound optimized over
  all possible $p$ as the optimized \LC bound.  Note that the special case of
  $p = 1$ is also known as the Chau bound~\cite{Chau13}.  Moreover, by
  adapting from Ref.~\cite{Chau23}, we know that
  \begin{equation}
   A_1 = \sin\varphi(0) .
   \label{E:A_1_value}
  \end{equation}
\end{itemize}

 Actually, all but the \LC bound above can be saturated in the sense that for
 each $p \in (0,2]$ and for any $\epsilon \in [0,1]$, there exists a pair of
 time-independent Hamiltonian and initial pure quantum state that attains the
 bound.  Whereas for the \LC bound, it can be saturated for all $\epsilon \in
 [0,1]$ when $p \le \pi/2$.  However, it is not clear if it can be saturated
 for $p \in (\pi/2,\pi]$~\cite{LC13}.  We give a negative answer to this
 question in Sec.~\ref{Sec:bound} by showing that the \LC bound may not be
 tight when $p \in (\pi/2,\pi]$ for a general $\epsilon$.  Lastly, we remark
 that all these bounds are complementary in the sense that each of the bounds
 cannot be reduced to another if we are not allowed to optimize the \LZ and
 \LC bounds over the parameter $p$.

 Observe that all the above \QSL{s} are in the form of a product of two terms.
 One is a function of a certain energy moment of the quantum state only.  The
 other one is a function of the fidelity and perhaps also the parameter $p$
 only.  We will contrast this feature when we discuss an efficient algorithm
 in computing the \CZ bound in Sec.~\ref{Sec:numeric} below.

 Last but not least, there is a different type of \QSL reported in the
 literature known as the exact \QSL{s}.  In particular, Pati \emph{et al.}
 proved an exact evolution time expression for finite-dimensional quantum
 systems as well as systems evolving under an Hamiltonian $H$ with $H^2 = I$.
 These expressions are valid for time-dependent as well as time-independent
 Hamiltonians under a technical condition to be discussed below.  Moreover,
 such an exact relation becomes an inequality for infinite-dimensional
 systems~\cite{PMSB23}.  So why extending and strengthening other more
 ``conventional'' \QSL{s} if the actual evolution time is know?  Here we
 answer this question by discussing the case of a two-dimensional quantum
 state as the starting point.

 For a two-dimensional initial state $\ket{\Psi(0)}$ evolving under a
 time-independent Hamiltonian, Pati \emph{et al.} proved that evolution time
 $\tau$ satisfies~\cite{PMSB23}
 \begin{equation}
  \frac{\tau}{\hbar} = \frac{\cos^{-1}(\sqrt{\epsilon})}{\llangle \Delta
  H^\text{nc} \rrangle_\tau}
  \label{E:exact_result}
 \end{equation}
 where
 \begin{equation}
  \llangle \Delta H^\text{nc} \rrangle_\tau = -\frac{\hbar}{2\tau} \int_0^\tau
  \frac{d p_t/dt}{\sqrt{p_t (1-p_t)}} \ dt
  \label{E:Delta_H_nc_def}
 \end{equation}
 provided that $p_t \equiv \left| \braket{\Psi(0)|\Psi(t)} \right|^2$ is
 monotonic decreasing for $t \in [0,\tau]$.  Note that even though the
 Hamiltonian is time-independent, computing $\llangle \Delta H^\text{nc}
 \rrangle$ requires integration over time.  More importantly, since $p_t$ is
 the fidelity square between $\ket{\Psi(0)}$ and $\ket{\Psi(t)}$, knowing
 $p_t$ at all times is equivalent to knowing $\ket{\Psi(0)}$ and the
 Hamiltonian.  That is to say, the information needed to determine the
 evolution time is encoded in $p_t$.  Therefore, the $\tau$ in
 Eq.~\eqref{E:exact_result} is an equality because it comes from tracing the
 time evolution of $\ket{\Psi(0)}$.  Actually, explicitly integrating the
 R.H.S. of Eq.~\eqref{E:Delta_H_nc_def}, Pati \emph{et al.}
 obtains~\cite{PMSB23}
\begin{equation}
 \llangle \Delta H^\text{nc} \rrangle_\tau = \frac{\hbar
 \cos^{-1}(\sqrt{p_\tau})}{2} .
 \label{E:Delta_H_nc_2d}
\end{equation}
 Thus, $2\llangle \Delta H^\text{nc} \rrangle_\tau / \hbar$ is the Bures angle
 between $\ket{\Psi(0)}$ and $\ket{\Psi(\tau)}$ in disguise.  This explains
 why a monotonically decreasing $p_t$ is needed to arrive at
 Eq.~\eqref{E:exact_result}.  In this regard, it is more efficient and
 accurate as well as conceptually simpler to compute the evolution time $\tau$
 by finding the smallest non-negative root of the equation $p_t = \epsilon^2$.
 By the same token, an exact \QSL is simply an alternative, possibly
 mathematically pleasing and inspiring, form of expressing the evolution time
 given a complete description of the initial state and evolution Hamiltonian.
 This argument holds for all exact \QSL{s}.

 To conclude, the exact \QSL is equivalent to computing the actual evolution
 time given complete information on the Hamiltonian and the initial state.
 Obviously, its calculation is difficult in general.  In contrast, this type
 of bounds are markedly different from the conventional \QSL{s} that give
 lower time bounds that are relatively easy to calculate based on partial
 information on the system (such as $\braket{E-E_0}$ or $\Delta E$).  Since
 the study of exact \QSL is conceptually different from those of conventional
 \QSL{s}, we do not consider this type of exact \QSL{s} in this paper.

\section{Auxiliary Results}
\label{Sec:auxiliary}
 We need the following auxiliary results whose proofs can be found in the
 Appendix.  Lemma~\ref{Lem:cos_inequality} below generalizes Lemma~1 in
 Ref.~\cite{Chau13} as well as Lemma~1 and Corollary~1 in Ref.~\cite{Chau23}.
 Its proof is based partly on those of Lemma~1 and Corollary~1 in
 Ref.~\cite{Chau23}.  Lemma~\ref{Lem:opt_LC} was first reported in one of the
 authors' capstone project report~\cite{zeng_thesis}.  In addition,
 Corollary~\ref{Cor:varphi_trend} extends Corollary~1 in Ref.~\cite{Chau23}.

\begin{lemma}
 Suppose $(p,\theta)\in {\mathcal R} \equiv (0,1] \times (-\pi,\pi/2] \cup
 (1,2] \times (-\pi,0]$, then
 \begin{equation}
  \cos x \ge \cos\theta - A_{p,\theta} (x-\theta)^p
  \label{E:cos_inequality}
 \end{equation}
 for all $x > \theta$, where
 \begin{equation}
  A_{p,\theta} = \max_{x \in [|\theta|,\pi)} \frac{\cos\theta - \cos
  x}{(x-\theta)^p} > 0 .
  \label{E:A_p_theta_def}
 \end{equation}
 (Note that for the case of $x = \theta$, the R.H.S. of
 Eq.~\eqref{E:A_p_theta_def} is regarded as the limit $x\to \theta^+$.  It
 exists for $(p,\theta) \in {\mathcal R}$).  Moreover, the $x$ that maximizes
 the R.H.S. of Eq.~\eqref{E:A_p_theta_def} is unique.  By writing this unique
 $x$ as $\varphi_p \equiv \varphi_p(\theta)$, then $A_{p,\theta}$ can be
 expressed as
 \begin{equation}
  A_{p,\theta} = \frac{\sin\varphi_p(\theta)}{p [\varphi_p(\theta) -
  \theta]^{p-1}} .
  \label{E:A_p_theta_expression}
 \end{equation}
 In fact, $\varphi_p$ is also the unique solution of the equation
 \begin{equation}
  f_{p,\theta}(\varphi_p) \equiv p (\cos\varphi_p - \cos\theta) + (\varphi_p -
  \theta) \sin\varphi_p = 0
  \label{E:varphi_def}
 \end{equation}
 in the interval $(|\theta|,\pi)$ if $(p,\theta) \ne (2,0)$.  Whereas if
 $(p,\theta) = (2,0)$, then $\varphi_2(0) = 0$, which is the unique solution
 of Eq.~\eqref{E:varphi_def} in the interval $[|\theta|,\pi)$.  Furthermore,
 for $p \in (0,1]$, the maximum in the R.H.S. of Eq.~\eqref{E:A_p_theta_def}
 can be taken over $x\in [\max(|\theta|,\pi/2),\pi)$.  Moreover, in the domain
 $x \in [\theta,+\infty)$, Inequality~\eqref{E:cos_inequality} becomes an
 equality if and only if $x = \theta$ or $\varphi_p(\theta)$.  Lastly,
 $\varphi_p(\theta)$ is a simple root of Eq.~\eqref{E:varphi_def} if
 $(p,\theta) \ne (2,0)$; and it is a root of order~4 otherwise.
 \label{Lem:cos_inequality}
\end{lemma}

 The following corollary follows directly from applying
 Lemma~\ref{Lem:cos_inequality} to $x \ge \theta$ and to $y = -x \ge \theta$
 separately.

\begin{corollary}
 For $p\in (0,1]$ and $\theta\in [-\pi/2,\pi/2]$, we have
 \begin{align}
  \cos x \ge& \cos\theta - \cf_{x \ge \theta} A_{p,\theta}^+ |x - \theta|^p -
   \cf_{x<\theta} A_{p,\theta}^- |\theta - x|^p \nonumber \\
  \equiv& \cos\theta - \cf_{x \ge \theta} A_{p,\theta} |x - \theta|^p -
   \cf_{x<\theta} A_{p,-\theta} |\theta - x|^p \nonumber \\
  =& \cos\theta - \frac{\cf_{x \ge \theta} |x - \theta|^p
   \sin\varphi_p^+(\theta)}{p [\varphi_p^+(\theta) - \theta]^{p-1}} +
   \frac{\cf_{x<\theta} |\theta - x|^p \sin\varphi_p^-(\theta)}{p [\theta -
   \varphi_p^-(\theta)]^{p-1}}
  \label{E:extended_cos_inequality}
 \end{align}
 for all $x\in {\mathbb R}$, with equality holds if and only if $x = \theta$,
 $x = \varphi_p^+(\theta) \equiv \varphi_p(\theta)$ or $x =
 \varphi_p^-(\theta) \equiv -\varphi_p(-\theta)$.  Here
 $\varphi_p(\pm\theta)$ are (unique) solutions of Eq.~\eqref{E:varphi_def} in
 the interval $(|\theta|,\pi)$,
 \begin{equation}
  \cf_{x \ge \theta} =
  \begin{cases}
   1 & \text{if~} x \ge \theta , \\
   0 & \text{otherwise} .
  \end{cases}
  \label{E:char_func_def}
 \end{equation}
 and $\cf_{x < \theta}$ is similarly defined.  Moreover, for the special case
 of $\theta = 0$, Inequality~\eqref{E:extended_cos_inequality} holds also for
 $p\in (1,2]$.  That is to say, for $p\in (0,2]$,
 \begin{equation}
  \cos x \ge 1 - \frac{|x|^p \sin\varphi_p^+(0)}{p \varphi_p^+(0)^{p-1}}
  \label{E:extended_cos_inequality_special}
 \end{equation}
 with equality holds if and only if $x = 0,\varphi_p^\pm(0)$.
 \label{Cor:extended_cos_inequality}
\end{corollary}

\begin{figure}[t]
 \centering\includegraphics[width=7.5cm]{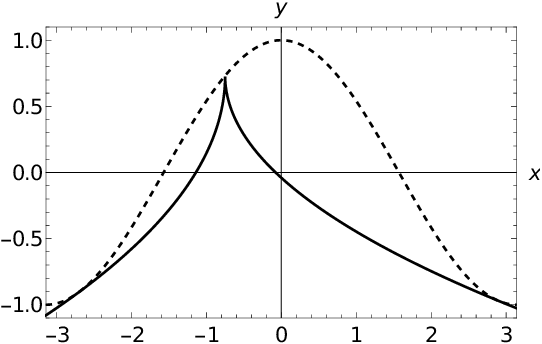}
 \caption{The dashed curve is $y = \cos x$.  The solid curve corresponds to
  the R.H.S. of Inequality~\eqref{E:extended_cos_inequality} for $\theta =
  -0.75$ and $p = 0.5$.  Clearly, the two curves meets at three distinct
  points with two of them being points of tangency.
  \label{F:extended_cos_inequality}}
\end{figure}

\begin{remark}
 Fig.~\ref{F:extended_cos_inequality} plots the L.H.S. and R.H.S. of
 Inequality~\eqref{E:extended_cos_inequality} and highlights its geometric
 meaning.  Clearly, $x = \varphi_p^+(\theta) = \varphi_p(\theta)$ is the
 (unique) point in $[|\theta|,\pi)$ where the curves $\cos x$ and $\cos\theta
 -A_{p,\theta} (x-\theta)^p$ meet tangentially.  Similarly, we may also
 interpret $x = -\varphi_p(-\theta)$ as the (unique) point in
 $(-\pi,-|\theta|]$ where the curves $\cos x$ and $\cos\theta - A_{p,-\theta}
 (\theta - x)^p$ meets tangentially.  Consequently, $\varphi_p^-(\theta) =
 -\varphi_p(-\theta)$ must be the unique solution of Eq.~\eqref{E:varphi_def}
 in the interval $(-\pi,-|\theta|]$ if $(p,\theta) \ne (2,0)$.
 \label{Rem:geometric_meaning}
\end{remark}

 Since Inequality~\eqref{E:extended_cos_inequality} plays a key role in this
 study, we use the notations $\varphi_p^\pm(\theta)$ and $A_{p,\theta}^\pm$
 instead of $\varphi_p(\theta)$, $-\varphi_p(-\theta)$, $A_{p,\theta}$ and
 $A_{p,-\theta}$ from now on except possibly for the case when $\theta = 0$.

\begin{corollary}
 Suppose $0 < p \le 1$ and $\theta \in [-\pi/2,\pi/2]$.  Then,
 $\varphi_p^+(\theta) \pm \theta \in [0,3\pi/2)$ and $\varphi_p^-(\theta) \pm
 \theta \in (-3\pi/2,0]$.  Besides, $\varphi_p^\pm(\theta) - \theta$ are
 strictly decreasing functions of $\theta$ and $\varphi_p^\pm(\theta) +
 \theta$ are strictly increasing functions of $\theta$.
 \label{Cor:varphi_trend}
\end{corollary}

\begin{remark}
 Ref.~\cite{Chau23} further showed that for $p = 1$, $\varphi_p^+(\theta)$ is
 an increasing function of $\theta$ in the domain $[-\pi/2,\pi/2]$.
 Nevertheless, this is not true for $0 < p < 1$.
 \label{Rem:varphi_trend}
\end{remark}

\begin{lemma}
 Let
 \begin{equation}
  h(x) =
  \begin{cases}
   x\cot(x/2) & \text{for~} x \in (0,\pi), \\
   \\
   \displaystyle \lim_{x\to 0^+} h(x) = 2 & \text{for~} x = 0 .
  \end{cases}
  \label{E:h_def}
 \end{equation}
 Then, $h$ is a strictly decreasing function in $[0,\pi]$.  Moreover, $h
 \colon [0,\pi] \mapsto [0,2]$ is a homeomorphism.
 \label{Lem:x_cotx_div_2}
\end{lemma}

\begin{lemma}
 Let $\epsilon \in [0,1]$.  Then
 \begin{equation}
  \sup_{x \in (0,\pi)} \left[ x \left( \frac{1 - \sqrt{\epsilon}}{2 \sin^2
  \frac{x}{2}} \right)^\frac{1}{x\cot(x/2)} \right] = \cos^{-1}
  (\sqrt{\epsilon}) .
  \label{E:opt_LC_aux}
 \end{equation}
 In fact, this supremum is attained at $x = \cos^{-1}(\sqrt{\epsilon})$.  In
 other words, this is actually a maximum.
 \label{Lem:opt_LC}
\end{lemma}

\begin{lemma}
 Let $w \colon [a,b] \to {\mathbb R}$ be a doubly differentiable function,
 $w$ is continuous on $[a,b]$, $w''$ is continuous on $(a,b)$.  Suppose
 further that $w(a) \ge 0$, $w(b) \le 0$, $w'(b) < 0$ and $w''(x) \le 0$ for
 all $x\in [a,b]$.  Then by choosing $b$ as the initial guess, Newton's method
 of finding a root of $w(x) = 0$ in $[a,b]$ always converges.  Surely, the
 convergence is quadratic if this root is simple.
 \label{Lem:Newton_method}
\end{lemma}

\section{The \CZ Bound}
\label{Sec:bound}
 In what follows, we adopt the following notations.  We write a
 time-independent Hamiltonian as a formal sum of its energy eigenvectors,
 namely, $\sum_j E_j \ket{E_j} \bra{E_j}$ with $\ket{E_j}$'s being the energy
 eigenstate of the Hamiltonian.  In addition, we write a normalized initial
 pure state $\ket{\Psi(0)}$ as $\sum_j a_j \ket{E_j}$ with $\sum_j |a_j|^2 =
 1$.  From now on, unless otherwise stated, a Hamiltonian and a quantum state
 in this paper are time-independent and pure, respectively.

\begin{theorem}[\CZ bound]
 The evolution time $\tau$ needed for any quantum state to evolve to another
 state whose fidelity between them is $\epsilon$ under a (time-independent)
 Hamiltonian satisfies the inequality
 \begin{subequations}
 \label{E:CZ_bound_basic}
 \begin{equation}
  \frac{\tau}{\hbar} \ge \frac{T_p(\epsilon)}{\hbar} \equiv
  \max_{\substack{|\theta| \le \cos^{-1}(\sqrt{\epsilon}), \\ E_r \in
  {\mathbb R}}} \left[ \frac{\cos\theta - \sqrt{\epsilon}}{A_{p,\theta}^+
  \eplus{(E-E_r)}{p} + A_{p,\theta}^- \eplus{(E_r-E)}{p}} \right]^\frac{1}{p}
  \label{E:CZ_bound_basic_generic}
 \end{equation}
  for $p \in (0,1]$ and
 \begin{equation}
  \frac{\tau}{\hbar} \ge \frac{T_p(\epsilon)}{\hbar} \equiv \max_{E_r \in
  {\mathbb R}} \left[ \frac{1 - \sqrt{\epsilon}}{A_{p,0} \braket{|E-E_r|^p}}
  \right]^\frac{1}{p}
  \label{E:CZ_bound_basic_p2}
 \end{equation}
 \end{subequations}
 for $p \in (1,2]$ provided that
 \begin{equation}
  \eplus{(E-E_r)}{p} = \sum_{j\colon E_j > E_r} |a_j|^2 |E_j - E_r|^p
  \label{E:Eplus_def}
 \end{equation}
 and the similarly defined $\eplus{(E_r-E)}{p}$ exist.  (Note that
 $\eplus{(E-E_r)}{p}$ and $\eplus{(E_r-E)}{p}$ are the expected measurement
 results of valid observables.  Physically, they relate to the $p$th moment of
 the absolute value of energy and the $p$th signed moment of the absolute
 value of energy of the initial state $\ket{\Psi(0)}$ via the relations
 \begin{subequations}
 \label{E:energy_moments}
 \begin{equation}
  \braket{|E-E_r|^p} = \eplus{(E-E_r)}{p} + \eplus{(E_r-E)}{p}
  \label{E:unsigned_moment}
 \end{equation}
 and
 \begin{equation}
  \braket{\sgn(E-E_r) |E-E_r|^p} = \eplus{(E-E_r)}{p} - \eplus{(E_r-E)}{p} .
  \label{E:signed_moment}
 \end{equation}
 \end{subequations}
 Note further that the denominator of the R.H.S. of
 Inequality~\eqref{E:CZ_bound_basic} vanishes if the initial state is an
 energy eigenstate of the Hamiltonian.  In this case,
 Inequality~\eqref{E:CZ_bound_basic} still holds if one interprets its R.H.S.
 as $0$ if $\epsilon = 1$ and $+\infty$ otherwise.)  The necessary and
 sufficient conditions for the Inequality~\eqref{E:CZ_bound_basic} to be
 saturated by a pair of Hamiltonian and initial quantum state are as follows.
 Let $(\theta_\opt, E_{r,\opt})$ be one of the possible pairs of values of
 $(\theta, E_r)$ that maximize the R.H.S. of
 Inequality~\eqref{E:CZ_bound_basic}.  (We shall show in
 Sec.~\ref{Subsec:numeric_E_r} that $E_{r,\opt}$ is unique if $1 < p \le 2$.
 For the other cases, $E_{r,\opt}$ may not be unique.  Also, we shall show in
 Theorem~\ref{Thrm:CZ_bound_computation} of Sec.~\ref{Subsec:numeric_theta}
 that $\theta_\opt$ is unique for any given $E_{r,\opt}$.)  Let us denote
 $\varphi_p^\pm(\theta_\opt)$ by $\varphi_{p,\opt}^\pm$.  Then, we have the
 followings.
 \begin{itemize}
  \item For $p\in (0,2]$ and $\epsilon = 1$ (and hence $\theta_\opt = 0$),
   Inequality~\eqref{E:CZ_bound_basic} can be saturated by any Hamiltonian and
   initial quantum state pair.
  \item For $p = 2$ and $\epsilon < 1$, Inequality~\eqref{E:CZ_bound_basic}
   cannot be saturated.
  \item For $p \in (0,1]$ and $\theta_\opt \in
   [-\cos^{-1}(\sqrt{\epsilon}),\cos^{-1}(\sqrt{\epsilon})]$,
   Inequality~\eqref{E:CZ_bound_basic} is saturated if and only if the
   (normalized) initial pure quantum state, when expressed in terms of energy
   eigenvectors of the corresponding Hamiltonian, equals
   \begin{equation}
    \ket{\Psi(0)} = a_+ \ket{E_+} + a_r \ket{E_{r,\opt}} + a_- \ket{E_-} .
    \label{E:Psi_optimal_form}
   \end{equation}
   Here $E_- < E_{r,\opt} < E_+$, $(E_+ - E_{r,\opt}) : (E_- - E_{r,\opt}) =
   (\varphi_{p,\opt}^+ - \theta_\opt) : (\varphi_{p,\opt}^- - \theta_\opt)$
   and
   \begin{align}
    &
    \begin{bmatrix}
     |a_+|^2 \\
     |a_r|^2 \\
     |a_-|^2
    \end{bmatrix}
    \nonumber \\
    ={}& \frac{1}{2\sin \frac{\varphi_{p,\opt}^+ - \varphi_{p,\opt}^-}{2} \sin
    \frac{\varphi_{p,\opt}^+ - \theta_\opt}{2} \sin \frac{\theta_\opt -
    \varphi_{p,\opt}^-}{2}}
    \begin{bmatrix}
     \sin \frac{\theta_\opt - \varphi_{p,\opt}^-}{2} \left( \cos
      \frac{\theta_\opt - \varphi_{p,\opt}^-}{2} - \sqrt{\epsilon} \, \cos
      \frac{\theta_\opt + \varphi_{p,\opt}^-}{2} \right) \\
     -\sin \frac{\varphi_{p,\opt}^+ - \varphi_{p,\opt}^-}{2} \left( \cos
      \frac{\varphi_{p,\opt}^+ - \varphi_{p,\opt}^-}{2} - \sqrt{\epsilon}
      \, \cos \frac{\varphi_{p,\opt}^+ + \varphi_{p,\opt}^-}{2} \right) \\
     \sin \frac{\varphi_{p,\opt}^+ - \theta_\opt}{2} \left( \cos
      \frac{\varphi_{p,\opt}^+ - \theta_\opt}{2} - \sqrt{\epsilon} \, \cos
      \frac{\varphi_{p,\opt}^+ + \theta_\opt}{2} \right)
    \end{bmatrix}
    .
    \label{E:prob_ampl_req}
   \end{align}
   (To be more explicit, the corresponding optimizing Hamiltonian is in the
   form $H = E_+ \ket{E_+}\bra{E_+} + E_r \ket{E_r}\bra{E_r} + E_-
   \ket{E_-}\bra{E_-} + H^\perp$ where $H^\perp$ is a Hamiltonian whose
   support is orthogonal to $\ket{E_+}$, $\ket{E_r}$ and $\ket{E_-}$.)
  \item For $p \in (1,2)$ (and hence $\theta_\opt = 0$),
   Inequality~\eqref{E:CZ_bound_basic} is saturated if and only if
   Eqs.~\eqref{E:Psi_optimal_form} and~\eqref{E:prob_ampl_req} hold.
   Furthermore,
   \begin{equation}
    \sqrt{\epsilon} \ge \cos \varphi_{p,\opt}^+ .
    \label{E:prob_ampl_req_varphi}
   \end{equation}
 \end{itemize}
 \label{Thrm:CZ_bound}
\end{theorem}
\begin{proof}
 This proof follows the basic ideas used in Ref.~\cite{Chau13} as well as the
 proof of Theorem~1 in Ref.~\cite{Chau23}.  The part on the saturation of
 Inequality~\eqref{E:CZ_bound_basic} extends the proof of Theorem~2 in
 Ref.~\cite{Chau23}.

 Since the fidelity between two pure states does not decrease under partial
 trace and the R.H.S. of Inequality~\eqref{E:CZ_bound_basic} is a decreasing
 function of fidelity, we only need to consider pure state evolution in the
 extended Hilbert space in our proof~\cite{GLM03}.

 To save space, we only prove the case of $p \in (0,1]$ here.  The case of $p
 \in (1,2]$ can be proven in the same way.  Details are left to interested
 readers.  Using the notations stated at the beginning of this Section and by
 Corollary~\ref{Cor:extended_cos_inequality}, we get
\begin{align}
 \sqrt{\epsilon} ={}& \left| \braket{\Psi(0) | \Psi(\tau)} \right| \ge
  \Re \left( \braket{\Psi(0) | \Psi(\tau)} e^{i E_r \tau / \hbar} e^{-i\theta}
  \right) = \sum_j |a_j|^2 \cos \left[ \frac{(E_j - E_r) \tau}{\hbar} + \theta
  \right] \nonumber \\
 \ge{}& \sum_j |a_j|^2 \left[ \cos\theta - \left( \cf_{E_j \ge E_r}
  A_{p,\theta}^+ + \cf_{E_j < E_r} A_{p,\theta}^- \right) \left| \frac{(E_j -
  E_r) \tau}{\hbar} \right|^p \right] \nonumber \\
 ={}& \cos\theta - \left\{ A_{p,\theta}^+ \eplus{(E-E_r)}{p} + A_{p,\theta}^-
  \eplus{(E_r-E)}{p} \right\} \left( \frac{\tau}{\hbar} \right)^p
 \label{E:fidelity1}
\end{align}
 for any $|\theta| \le \pi/2$ and for any reference energy level $E_r$.  Here
 $\Re(\cdot)$ denotes the real part of its argument.

 We assume that the coefficient of the $(\tau / \hbar)^p$ term is positive.
 (If not, $\ket{\Psi(0)} = \ket{E_r}$ up to an irrelevant phase and hence
 Eq.~\eqref{E:CZ_bound_basic} is trivially true according to the convention
 stated in this Theorem.)  Then, Inequality~\eqref{E:fidelity1} can be
 rewritten as
\begin{equation}
 \left( \frac{\tau}{\hbar} \right)^{p} \ge \frac{\cos\theta -
 \sqrt{\epsilon}}{A_{p,\theta}^+ \eplus{(E-E_r)}{p} + A_{p,\theta}^-
 \eplus{(E_r-E)}{p}} .
 \label{E:CZ_bound_basic_generic_alt}
\end{equation}
 Clearly, Inequality~\eqref{E:CZ_bound_basic_generic_alt} gives meaningful
 constraint on the evolution time $\tau$ if $\cos\theta \ge \sqrt{\epsilon}$.
 Furthermore, the denominator of the R.H.S. of
 Inequality~\eqref{E:CZ_bound_basic_generic_alt} is a continuous non-negative
 function of $\theta$.  Besides, it is unbounded if $E_r \to \pm\infty$.  So,
 for fixed values of $p, \epsilon$ and $\theta$, there is an $E_r$ that
 minimizes the denominator of the R.H.S. of
 Inequality~\eqref{E:CZ_bound_basic_generic_alt} even though such an $E_r$
 need not be unique in general.  Hence, we can maximize the R.H.S. of
 Inequality~\eqref{E:CZ_bound_basic_generic_alt} by minimizing its denominator
 over $E_r$ as well as maximizing over $\theta \in
 [-\cos^{-1}(\sqrt{\epsilon}),\cos^{-1}(\sqrt{\epsilon})]$.  (Obviously, the
 order of minimization and maximization does not affect the final outcome.)
 This gives Inequality~\eqref{E:CZ_bound_basic_generic}.

 We now find the set of all (time-independent) Hamiltonian and initial (pure)
 quantum state pairs that saturate Inequality~\eqref{E:CZ_bound_basic}.  The
 necessary and sufficient conditions for the case of $\epsilon = 1$ are
 obvious as $\tau = 0$.  So, we assume that $\epsilon \in [0,1)$ from now on.
 The necessary and sufficient conditions for the first line of
 Inequality~\eqref{E:fidelity1} to be an equality are that $\Re
 (\braket{\Psi(0) | \Psi(\tau)} e^{i E_r \tau / \hbar} e^{-i\theta}) =
 \sqrt{\epsilon} \ge 0$ and $\Im (\braket{\Psi(0) | \Psi(\tau)} e^{i E_r \tau
 / \hbar} e^{-i\theta}) = 0$, where $\Im(\cdot)$ denotes the imaginary part of
 the argument.  In addition, from Corollary~\ref{Cor:extended_cos_inequality},
 the second line of Inequality~\eqref{E:fidelity1} is an equality if and only
 if $\braket{E_j|\Psi(0)} = 0$ for all $(E_j - E_r)\tau/\hbar \notin \{ 0,
 \varphi_p^{\pm} - \theta \}$.  As a result, the normalized initial state
 $\ket{\Psi(0)}$ must be in the form of Eq.~\eqref{E:Psi_optimal_form} (with
 all the ``opt'' subscripts removed).

 Now we have enough information to discuss the necessary and sufficient
 conditions for saturation of Inequality~\eqref{E:CZ_bound_basic} for the case
 of $p = 2$ and $\epsilon < 1$.  According to Lemma~\ref{Lem:cos_inequality},
 $\varphi_2^\pm(0) = 0$.  Thus, the saturating state $\ket{\Psi(0)}$ in
 Eq.~\eqref{E:Psi_optimal_form} is simply $\ket{E_r}$ up to a global phase.
 Being an energy eigenstate, $\ket{\Psi(0)}$ does not evolve with time.  That
 is why Inequality~\eqref{E:CZ_bound_basic} cannot be saturated in this
 situation.

 Let us continue our saturation condition analysis for the remaining cases.
 From our discussions so far, the normalization of $\ket{\Psi(0)}$ together
 with the requirements on the real and imaginary parts of
 $\braket{\Psi(0)|\Psi(\tau)} e^{i E_r \tau/\hbar} e^{-i\theta}$, we have
 \begin{equation}
  \begin{bmatrix}
   1 & 1 & 1 \\
   \cos\varphi_p^+ & \cos\theta & \cos\varphi_p^- \\
   \sin\varphi_p^+ & \sin\theta & \sin\varphi_p^-
  \end{bmatrix}
  \begin{bmatrix}
   |a_+|^2 \\
   |a_r|^2 \\
   |a_-|^2
  \end{bmatrix}
  =
  \begin{bmatrix}
   1 \\
   \sqrt{\epsilon} \\
   0
  \end{bmatrix}
  .
  \label{E:prob_ampl}
 \end{equation}
 Recall from Lemma~\ref{Lem:cos_inequality} that $-\pi < \varphi_p^- < \theta
 < \varphi_p^+ < \pi$ for $(p,\theta) \in {\mathcal R} \setminus (2,0)$.
 Hence, Eq.~\eqref{E:prob_ampl} has a unique solution.  By solving
 Eq.~\eqref{E:prob_ampl} and then substituting $\theta \in
 [-\cos^{-1}(\sqrt{\epsilon}),\cos^{-1}(\sqrt{\epsilon})]$ by its optimal
 value $\theta_\opt$ obtained in the R.H.S. of
 Inequality~\eqref{E:CZ_bound_basic}, we conclude that $a_+,a_r$ and $a_-$
 must obey Eq.~\eqref{E:prob_ampl_req}.

 Finally, a valid $\ket{\Psi(0)}$ requires $|a_+|^2, |a_r|^2$ and $|a_-|^2$ to
 be non-negative.  In the case of $p\in (0,1]$ and $\theta_\opt \ge 0$,
 Lemma~\ref{Lem:cos_inequality} and~Corollary~\ref{Cor:varphi_trend} imply
 that $\varphi_{p,\opt}^+ - \varphi_{p,\opt}^- \in [\pi,2\pi)$,
 $\varphi_{p,\opt}^+ + \varphi_{p,\opt}^- \in (-\pi/2,\pi/2)$,
 $\varphi_{p,\opt}^+ - \theta_\opt \in [0,\pi)$, $\varphi_{p,\opt}^+ +
 \theta_\opt \in [\pi/2,3\pi/2)$, $\theta_\opt - \varphi_{p,\opt}^- \in
 [\pi/2,3\pi/2)$ and $\theta_\opt + \varphi_{p,\opt}^- \in (-\pi,0]$.  So,
 from Eq.~\eqref{E:prob_ampl_req}, $|a_r|^2 \ge 0$.  As $(\varphi_{p,\opt}^+ +
 \theta_\opt)/2 = (\varphi_{p,\opt}^+ - \theta_\opt)/2 + \theta_\opt$, we
 conclude from the ranges of the arguments of the cosine function that $\cos
 [(\varphi_{p,\opt}^+ - \theta_\opt)/2] \ge \cos [(\varphi_{p,\opt}^+ +
 \theta_\opt)/2]$.  Hence, from Eq.~\eqref{E:prob_ampl_req}, $|a_-|^2 \ge 0$.
 Recall from Corollary~\ref{Cor:varphi_trend}, $\varphi_{p,\opt}^\pm -
 \theta_\opt$ are strictly decreasing functions of $\theta_\opt \in
 (-\cos^{-1}[\sqrt{\epsilon}],\cos^{-1}[\sqrt{\epsilon}]) \subset
 (-\pi/2,\pi/2)$.  Thus,
 \begin{equation}
  \frac{d}{d\theta_\opt} \left( \cos\frac{\theta_\opt - \varphi_{p,\opt}^-}{2}
  - \sqrt{\epsilon} \cos\frac{\theta_\opt + \varphi_{p,\opt}^-}{2} \right) < 0
  .
  \label{E:dcoscomplex_dtheta}
 \end{equation}
 By writing $\theta_\crit \equiv \cos^{-1}(\sqrt{\epsilon}) \in [0,\pi/2]$ and
 $\varphi_{p,\crit}^- \equiv \varphi_p^-(\theta_\crit)$, we know that
 \begin{align}
  \cos \frac{\theta_\opt - \varphi_{p,\opt}^-}{2} - \sqrt{\epsilon} \cos
   \frac{\theta_\opt + \varphi_{p,\opt}^-}{2} &\ge \cos \frac{\theta_\crit -
   \varphi_{p,\crit}^-}{2} - \cos\theta_\crit \cos \frac{\theta_\crit +
   \varphi_{p,\crit}^-}{2} \nonumber \\
  &= \cos \frac{\theta_\crit - \varphi_{p,\crit}^-}{2} - \frac{1}{2} \left(
   \cos \frac{3\theta_\crit + \varphi_{p,\crit}^-}{2} + \cos
   \frac{\theta_\crit - \varphi_{p,\crit}^-}{2} \right) \nonumber \\
  &= \sin\theta_\crit \sin \frac{\theta_\crit + \varphi_{p,\crit}^-}{2} \ge 0
   .
  \label{E:a+_positivity}
 \end{align}
 From Eq.~\eqref{E:prob_ampl_req}, we find that $|a_+|^2 \ge 0$.  Thus,
 $\ket{\Psi(0)}$ is a valid quantum state because $|a_\pm|^2, |a_r|^2 \ge 0$.

 The case of $(p,\theta) \in (0,1] \times [0,\pi/2]$ can be proven in a
 similar way.  Actually, showing $|a_+|^2, |a_r|^2 \ge 0$ is straightforward.
 Proving $|a_-|^2 \ge 0$ is more involved.  Its validity is due to
 \begin{align}
  \cos \frac{\varphi_{p,\opt}^+ - \theta_\opt}{2} - \sqrt{\epsilon} \cos
   \frac{\varphi_{p,\opt}^+ + \theta_\opt}{2} &\ge \cos
   \frac{\varphi_{p,\crit}^+ + \theta_\crit}{2} + \cos\theta_\crit \cos
   \frac{\varphi_{p,\crit}^+ - \theta_\crit}{2} \nonumber \\
  &= \frac{3}{2} \cos \frac{\varphi_{p,\crit}^+ + \theta_\crit}{2} +
   \frac{1}{2} \cos \frac{\varphi_{p,\crit}^+ - 3\theta_\crit}{2} \ge 0 ,
  \label{E:a-_positivity}
 \end{align}
 where $\varphi_{p,\crit}^+ \equiv \varphi_p^+(\theta_\crit)$.  Here we have
 used $\varphi_p^+ - 3\theta_\crit \in [-\pi,\pi)$ to arrive at the last
 inequality.

 The proof for the case of $p \in (1,2)$ follows similar logic though it is
 simpler as $\theta_\opt = 0$.  In fact, $|a_\pm|^2 \ge 0$ is trivially true
 in this case.  And the condition in Inequality~\eqref{E:prob_ampl_req_varphi}
 is due to the requirement that $|a_r|^2 \ge 0$.  Details are left to
 interested readers.
\end{proof}

\begin{remark}
 From the above proof, it is clear that by replacing $\theta_\opt$ with any
 $\theta \in [-\cos^{-1}(\sqrt{\epsilon}),\cos^{-1}(\sqrt{\epsilon})]$ as well
 as by replacing $E_{r,\opt}$ with any $E_r$, we obtain all Hamiltonian and
 initial quantum state pairs that saturate Inequality~\eqref{E:CZ_bound_basic}
 with the optimization over $\theta$ and $E_r$ removed.
 \label{Rem:general_theta}
\end{remark}

\begin{remark}
 By comparing Lemma~\ref{Lem:cos_inequality}, Theorem~\ref{Thrm:CZ_bound} and
 their proofs with the prior works summarized in Sec.~\ref{Sec:prior_art}, we
 observe that
 \begin{itemize}
  \item By putting $p = 1$ together with $E_r = E_0$ and optimizing
   Inequality~\eqref{E:CZ_bound_basic_generic_alt} over $\theta$ alone, we
   obtain the \ML bound.
  \item By putting $p = 1$, $E_r = E_{\max}$ and optimizing
   Inequality~\eqref{E:CZ_bound_basic_generic_alt} over $\theta$ alone, we
   arrive at the \dualML bound.
  \item When $\theta = \tan^{-1}(2p/\pi)$,
   Inequality~\eqref{E:extended_cos_inequality} becomes the key lemma in
   Ref.~\cite{LZ05}.  Hence, Inequality~\eqref{E:CZ_bound_basic_generic_alt}
   reduces to the \LZ bound for this $\theta$ when $E_r = E_0$.
  \item By setting $(p,\theta) = (1,0)$ and by optimizing
   Inequality~\eqref{E:CZ_bound_basic_generic_alt} over $E_r$ only, we get
   back the \Chau bound.
  \item By putting $\theta = 0$ and by optimizing over $E_r$ alone, the \CZ
   bound reduces to the \LC bound.  In this regard,
   Theorem~\ref{Thrm:CZ_bound} gives the necessary and sufficient conditions
   for saturating the \LC bound.  This fills the gap in Ref.~\cite{LC13}.
   Specifically, by going through the proof of Theorem~\ref{Thrm:CZ_bound},
   we find that the \LC bound can be saturated if and only if
   Inequality~\eqref{E:prob_ampl_req_varphi} holds with $\varphi_{p,\opt}^+ =
   \varphi_p^+(0)$.  In particular, the \LC bound can be saturated for all
   $\epsilon \in [0,1]$ if and only if $\varphi_p^+(0) \ge \pi/2$.  According
   to Eq.~\eqref{E:varphi_def}, $p = \varphi_p^+(0) \cot[\varphi_p^+(0)/2]$.
   As $p = \pi/2$ if $\varphi_p^+(0) = \pi/2$, Lemma~\ref{Lem:x_cotx_div_2}
   implies that the \LC bound can be saturated for all $\epsilon$ if and only
   if $p \le \pi/2$.  More importantly, if $p \ge \pi/2$, the \LC bound can
   be saturated if and only if $\epsilon \ge \epsilon_c \equiv
   \cos^2[h^{-1}(p)]$.  Clearly, $\epsilon_c$ is a strictly increasing
   function of $p$ in $[\pi/2,2]$.
  \item For $p = 2$, Inequality~\eqref{E:CZ_bound_basic_p2} agrees with the
   \MT bound up to $\Oh([1-\epsilon]^2)$ in the limit $\epsilon\to 1^-$.
   Moreover, for $\epsilon < 1$, Inequality~\eqref{E:CZ_bound_basic_p2} alone
   is not as powerful as the \MT bound.  Nevertheless,
   Theorem~\ref{Thrm:LC_bound_extension} below means that by optimizing over
   $p$ with $\theta = 0$, we recover the \MT bound.
 \end{itemize}
 These relations can be schematically represented in Fig.~\ref{F:relations}.
 \label{Rem:comparison}
\end{remark}

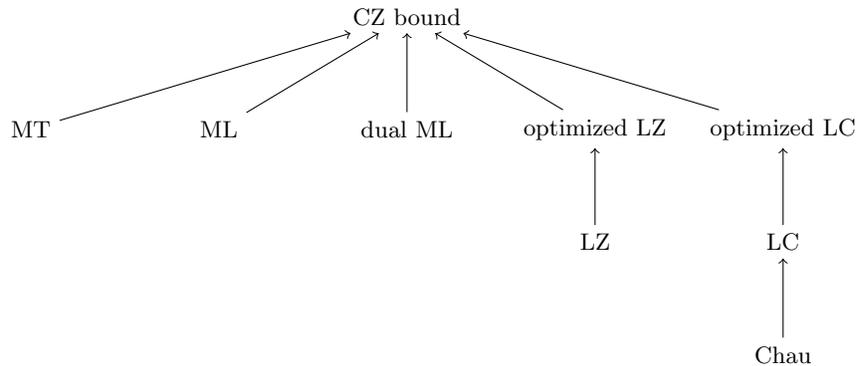
\begin{figure}[t]
 \begin{tikzpicture}[style={sibling distance = 2.5cm}]
  \node {\CZ bound}[<-]
   child { node{\MT} }
   child { node{\ML} }
   child { node{dual \ML} }
   child { node{optimized \LZ}[<-] child { node{\LZ} } }
   child { node{optimized \LC}[<-] child { node{\LC}[<-] child { node{\Chau} }
           } };
 \end{tikzpicture}
 \caption{\label{F:relations}
  Relations between various \QSL{s}.  Here $\text{A} \rightarrow \text{B}$
  means that $\text{A}$ is a special case of $\text{B}$.}
\end{figure}

 A closely related result is the following theorem.  Its special case for $p
 = 2$ was originally proven in the capstone project report of one of the
 authors~\cite{zeng_thesis}.

\begin{theorem}
 The \MT bound can be deduced by optimizing the \LC bound over $p\in (0,2)$.
 In fact,
 \begin{equation}
  \frac{\tau}{\hbar} \ge \max_{E_r}
  \frac{\cos^{-1}(\sqrt{\epsilon})}{\braket{|E-E_r|^p}^\frac{1}{p}}
  \label{E:LC_bound_extension}
 \end{equation}
 for any $p\in (0,2]$ provided that $\varphi_p(0) \le
 \cos^{-1}(\sqrt{\epsilon})$.
 \label{Thrm:LC_bound_extension}
\end{theorem}
\begin{proof}
 Note that $x^\lambda$ is a concave function for $0 < \lambda \le 1$ in the
 domain $x\ge 0$.  So, Jensen's inequality implies that
 $\braket{|E-E_r|^{q\lambda}} \ge \braket{|E-E_r|^q}^{\lambda}$ for any
 reference energy level $E_r$ provided that $\lambda \in (0,1]$.  Applying
 this inequality to Inequality~\eqref{E:fidelity1} with $\theta = 0$ (which is
 nothing but the \LC bound) and using Eq.~\eqref{E:unsigned_moment}, we get
 \begin{equation}
  \frac{\tau}{\hbar} \ge \left( \frac{1-\sqrt{\epsilon}}{A_{p,0}}
  \right)^\frac{1}{q} \frac{1}{\braket{|E-E_r|^{q\lambda}}^\frac{1}{q\lambda}}
  .
  \label{E:fidelity_extension1}
 \end{equation}
 for all $q\in (0,2)$.

 From Eq.~\eqref{E:varphi_def},
 \begin{equation}
  q = \varphi_q \cot\frac{\varphi_q}{2} \equiv \varphi_q(0) \cot 
  \frac{\varphi_q(0)}{2} .
  \label{E:q_for_theta=0}
 \end{equation}
 (Note that for this range of $q$, $\cot\varphi_q \ne 0$.  So
 Eq.~\eqref{E:q_for_theta=0} is well-defined.)  By
 Lemma~\ref{Lem:x_cotx_div_2}, Eq.~\eqref{E:q_for_theta=0} is a homeomorphism
 from $\varphi_q \in (0,\pi)$ to $q \in (0,2)$.

 By fixing $q\lambda$ to a certain $p \in (0,2]$ as well as by using
 Eqs.~\eqref{E:A_p_theta_expression} and~\eqref{E:q_for_theta=0}, we obtain
 \begin{equation}
  \frac{\tau}{\hbar} \ge \sup_{q\in (0,p)} \left(
  \frac{1-\sqrt{\epsilon}}{A_{q,0}} \right)^\frac{1}{q}
  \frac{1}{\braket{|E-E_r|^p}^\frac{1}{p}} = \sup_{q\in (0,p)} \left[
  \frac{(1-\sqrt{\epsilon}) \varphi_q^q}{2 \sin^2 \frac{\varphi_q}{2}}
  \right]^\frac{1}{\varphi_q \cot(\varphi_q/2)}
  \frac{1}{\braket{|E-E_r|^p}^\frac{1}{p}} .
  \label{E:fidelity_extension2}
 \end{equation}
 From the discussion in the last paragraph, we can replace $q\in (0,p)$ in the
 supremum in the last line of Inequality~\eqref{E:fidelity_extension2} by
 $\varphi_q \in (\varphi_p,\pi) \equiv (\varphi_p(0),\pi)$.  And from
 Lemma~\ref{Lem:opt_LC}, we conclude that the supremum of the R.H.S. of
 Inequality~\eqref{E:fidelity_extension2} is attained at $\varphi_q =
 \cos^{-1}(\sqrt{\epsilon})$ provided that $\varphi_p \le
 \cos^{-1}(\sqrt{\epsilon})$.  In this case, the first factor in the R.H.S. of
 Inequality~\eqref{E:fidelity_extension2} equals $\cos^{-1}(\sqrt{\epsilon})$.
 By maximizing the resultant inequality over $E_r$, we prove
 Inequality~\eqref{E:LC_bound_extension}.  We also mention on passing that in
 case $\varphi_p > \cos^{-1}(\sqrt{\epsilon})$, then from the proof of
 Lemma~\ref{Lem:opt_LC}, the supremum of
 Inequality~\eqref{E:fidelity_extension2} is attained at $\varphi_q =
 \varphi_p$.  In other words, we simply get back the \LC bound.
\end{proof}

\begin{remark}
 Recall from Lemma~\ref{Lem:opt_LC} that the optimized evolution time $\tau$
 in Theorem~\ref{Thrm:LC_bound_extension} is attained when $\varphi_p =
 \cos^{-1}(\sqrt{\epsilon})$.  By compound angle formula, the corresponding
 $p$ equals
 \begin{equation}
  p_\opt = \sqrt{\frac{1+\sqrt{\epsilon}}{1-\sqrt{\epsilon}}}
  \cos^{-1}(\sqrt{\epsilon}) .
  \label{E:opt_LC_p}
 \end{equation}
 \label{Rem:opt_LC_p}
\end{remark}

 An important consequence of Remark~\ref{Rem:comparison} is that we can
 further strengthen the \QSL reported in Theorem~\ref{Thrm:CZ_bound} as
 \begin{equation}
  \tau \ge \max_{p\in (0,2]} T_p(\epsilon) .
  \label{E:optimized_CZ_bound}
 \end{equation}
 By Theorem~\ref{Thrm:LC_bound_extension}, this is stronger than the combined
 \MT, \ML and \dualML bounds studied in Refs.~\cite{LT09,NAS22}.  From
 Remark~\ref{Rem:comparison}, the optimized \CZ bound generalizes the \MT,
 \ML, \dualML, \LZ and \LC bounds.  Nevertheless, the comparison between the
 combined \MT, \ML and \dualML bounds in Refs.~\cite{LT09,NAS22} with the
 optimized \CZ bound is not entirely fair.  This is because the combined bound
 in Refs.~\cite{LT09,NAS22} makes use of three numbers describing the initial
 quantum state, namely, $\braket{E - E_0}$, $\braket{E_{\max} - E}$ and
 $\Delta E$.  In contrast, Inequality~\eqref{E:optimized_CZ_bound} uses an
 infinite number of descriptions of the quantum state, each in the form
 $\eplus{(E-E_r)}{p}$ or $\eplus{(E_r-E)}{p}$.

 Along another line, one can generalize
 Corollary~\ref{Cor:extended_cos_inequality} by using two different values of
 $p$ --- one for $E > E_r$ and the other for $E < E_r$.  By the same argument
 as in the proof of Theorem~\ref{Thrm:CZ_bound}, we have the following \QSL.

\begin{theorem}
 Let $(p,q,\theta) \in (0,1] \times (0,1] \times [-\pi/2,\pi/2] \cup (0,1]
 \times (1,2] \times [0,\pi/2] \cup (1,2] \times (0,1] \times [-\pi/2,0] \cup
 (1,2] \times (1,2] \times \{ 0 \}$, we have
 \begin{equation}
  \sqrt{\epsilon} \ge \cos\theta - \left\{ A_{p,\theta}^+ \eplus{(E-E_r)}{p}
  \left( \frac{\tau}{\hbar} \right)^p + A_{q,\theta}^- \eplus{(E_r-E)}{q}
  \left( \frac{\tau}{\hbar} \right)^q \right\}
  \label{E:fidelity2}
 \end{equation}
 for all reference energy level $E_r$.  In particular, by putting $p = 2q$, we
 arrive at
 \begin{subequations}
 \label{E:CZ_bound_general}
 \begin{equation}
  \frac{\tau}{\hbar} \ge \max_{E_r \in {\mathbb R}} \left( \frac{ \left\{ 4
  A_{2q,\theta}^+ \eplus{(E-E_r)}{2q} (\cos\theta - \sqrt{\epsilon}) +
  (A_{q,\theta}^-)^2 \eplus{(E_r-E)}{q}^2 \right\}^{1/2} - A_{q,\theta}^-
  \eplus{(E_r-E)}{q}}{2 A_{2q,\theta}^+ \eplus{(E-E_r)}{2q}}
  \right)^\frac{1}{q}
  \label{E:CZ_bound_general1}
 \end{equation}
 for $(q,\theta) \in R \equiv (0,1/2] \times [-\pi/2,\pi/2] \cup (1/2,1]
 \times [0,\pi/2]$.  Similarly, the \QSL corresponding to the case of $q = 2p$
 is
 \begin{equation}
  \frac{\tau}{\hbar} \ge \max_{E_r \in {\mathbb R}} \left( \frac{ \left\{ 4
  A_{2p,\theta}^- \eplus{(E_r-E)}{2p} (\cos\theta - \sqrt{\epsilon}) +
  (A_{p,\theta}^+)^2 \eplus{(E-E_r)}{p}^2 \right\}^{1/2} - A_{p,\theta}^+
  \eplus{(E-E_r)}{p}}{2 A_{2p,\theta}^- \eplus{(E_r-E)}{2p}}
  \right)^\frac{1}{p} .
  \label{E:CZ_bound_general2}
 \end{equation}
 \end{subequations}
 \label{Thrm:CZ_bound_asy}
\end{theorem}

\begin{remark}
 In the proof of the saturation conditions for the \CZ bound in
 Theorem~\ref{Thrm:CZ_bound}, we do not use the property that the same $p$ is
 used for $\varphi_{p,\opt}^+$ and $\varphi_{p,\opt}^-$.  Therefore, simply by
 changing $\varphi_{p,\opt}^-$ to $\varphi_{q,\opt}^-$ in the saturation
 description part of Theorem~\ref{Thrm:CZ_bound}, we obtain the necessary and
 sufficient conditions for the \QSL expressed in
 Inequality~\eqref{E:fidelity2}.
 \label{Rem:CZ_bound_asy_saturation}
\end{remark}

 Although the \QSL reported in Inequality~\eqref{E:fidelity2} of
 Theorem~\ref{Thrm:CZ_bound_asy} is stronger than the \QSL in
 Theorem~\ref{Thrm:CZ_bound}, it has a few drawbacks.  First, an additional
 optimization over $q$ together with a numerical solution of the least
 positive root of Inequality~\eqref{E:fidelity2} with the inequality replaced
 by equality is required.  Second, the absence of a closed form for $\tau$
 makes finding $\tau$ more troublesome and the \QSL conceptually less
 appealing.  These two facts make the evaluation of the optimized version of
 this \QSL very time consuming.  One may consider the special cases such as
 Inequalities~\eqref{E:CZ_bound_general1} and~\eqref{E:CZ_bound_general2} in
 which the $p$ and $q$ are constrained.  Even though they are better than
 Inequality~\eqref{E:fidelity2} both in terms of the simplicity of the
 expression and computational tractability, they still lack the appeal of
 simplicity compared to the \CZ bound.  Moreover, as compared to the
 computational methods for the \CZ bound to be reported in
 Sec.~\ref{Sec:numeric}, it is not likely to obtain a similarly efficient
 method for the \QSL{s} stated in Theorem~\ref{Thrm:CZ_bound_asy}.   This is
 because these methods all start from an explicit expression of $\tau$.  Last
 but not least, we know from Remark~\ref{Rem:CZ_bound_asy_saturation} that
 Inequalities~\eqref{E:CZ_bound_general1} and~\eqref{E:CZ_bound_general2} are
 complementary to the \CZ bound.  This is because by writing in energy
 eigenbasis, those initial states saturating, say,
 Inequality~\eqref{E:CZ_bound_general1} must be in the form $a_+ \ket{E_+} +
 a_r \ket{E_{r,\opt}} + a_- \ket{E_-}$ with $(E_+ - E_{r,\opt}) : (E_{r,\opt}
 - E_-) = (\varphi_{2q,\opt}^+ - \theta_\opt) : (\varphi_{q,\opt}^- -
 \theta_\opt)$.  They are clearly different from those in
 Eq.~\eqref{E:Psi_optimal_form}.

\section{Accurate And Efficient Computation Of The \CZ Bound}
\label{Sec:numeric}
 Although the optimized \CZ bound in the form of
 Inequality~\eqref{E:CZ_bound_basic} in Theorem~\ref{Thrm:CZ_bound} is
 conceptually appealing, using it to calculate the optimized \CZ bound,
 namely, the \CZ bound optimized over $p$, $\theta$ and $E_r$, seems to be
 inefficient for the case of $p \in (0,1]$.  This is because one has to
 optimize over both $\theta$ and $E_r$ for any given $\epsilon$ (together with
 $\eplus{(E-E_r)}{p}$ and $\eplus{(E_r-E)}{p}$).  Furthermore, we have to
 further optimize over $p$ to obtain Inequality~\eqref{E:optimized_CZ_bound}.

 Here we report an accurate and efficient way to compute
 Inequalities~\eqref{E:CZ_bound_basic} given an oracle that accurately
 computes the minimized $p$th moment of the absolute value of energy of a
 quantum state and its corresponding $p$th signed moment of the absolute
 value of energy, or more precisely, $\braket{|E-E_{r,\opt}|^p} \equiv
 \min_{E_r\in {\mathbb R}} \braket{|E-E_r|^p}$ and $\braket{\sgn(E-E_{r,\opt})
 |E-E_{r,\opt}|^p}$.  And from Eq.~\eqref{E:energy_moments}, this is
 equivalent to having an oracle that accurately computes
 $\eplus{(E-E_{r,\opt})}{p}$ and $\eplus{(E_{r,\opt}-E)}{p}$.  The key
 observation is that $E_{r,\opt}$, the optimized value of $E_r$, is $\theta$
 independent.  So, we may first compute $E_{r,\opt}$.  Then, we evaluate
 $\eplus{(E-E_{r,\opt})}{p}$ and $\eplus{(E_{r,\opt}-E}{p}$ together with the
 optimized values of $\theta$ given $E_{r,\opt}$.  More explicitly,
 Inequality~\eqref{E:CZ_bound_basic_generic} can be rewritten as
\begin{align}
 \left( \frac{\tau}{\hbar} \right)^p &\ge \max_{|\theta| \le
  \cos^{-1}(\sqrt{\epsilon})} \left\{ \frac{\cos\theta -
  \sqrt{\epsilon}}{\left[ \mu^+(E_{r,\opt}) A_{p,\theta}^+ + \mu^-(E_{r,\opt})
  A_{p,\theta}^- \right] \displaystyle \min_{E_r \in {\mathbb R}}
  \braket{|E-E_r|^p}} \right\} \nonumber \\
 &= \max_{|\theta| \le \cos^{-1}(\sqrt{\epsilon})} \left\{ \frac{\cos\theta -
  \sqrt{\epsilon}}{\left[ \mu^+(E_{r,\opt}) A_{p,\theta}^+ + \mu^-(E_{r,\opt})
  A_{p,\theta}^- \right] \braket{|E-E_{r,\opt}|^p}} \right\} ,
 \label{E:CZ_bound_basic_generic_improved}
\end{align}
 where
\begin{equation}
 \eplus{(E-E_{r,\opt})}{p} : \eplus{(E_{r,\opt}-E)}{p} = \mu^+(E_{r,\opt}) :
 \mu^-(E_{r,\opt})
 \label{E:mu_pm_def}
\end{equation}
 with $\mu^+(E_{r,\opt}) + \mu^-(E_{r,\opt}) = 1$.

 Here we remark that Inequality~\eqref{E:CZ_bound_basic_generic_improved} is
 in almost the same form as all the \QSL{s} mentioned in
 Sec.~\ref{Sec:prior_art}.  More precisely, it is a product of a term
 depending only on the $p$th moment of the absolute value of energy of the
 quantum state optimized over the reference energy level with a term depending
 on $\epsilon$ and $\mu^\pm(E_{r,\opt})$.  In other words, the second term can
 be pre-computed and reused if $\mu^\pm(E_{r,\opt})$ are known in advance.

\subsection{On The Computation Of $\boldsymbol{\braket{|E-E_r|^p}}$ And
 $\boldsymbol{E_{r,\opt}}$}
\label{Subsec:numeric_E_r}
 Here we justify our assumption on the existence of an oracle to evaluate
 $\eplus{(E-E_{r,\opt})}{p}$ and $\eplus{(E_{r,\opt}-E)}{p}$ by showing that
 it can be replaced by computationally efficient and accurate algorithms in a
 number of useful situations.  In all cases, the idea is to first find
 $E_{r,\opt}$ and use it to obtain $\eplus{(E-E_{r,\opt})}{p}$ and
 $\eplus{(E_{r,\opt}-E)}{p}$ (or equivalently $\braket{|E-E_{r,\opt}|^p}$ and
 $\mu^\pm(E_{r,\opt})$).  The first case is when the normalized initial
 quantum state can be expressed in the form
\begin{equation}
 \ket{\Psi(0)} = \sum_{j = 1}^n a_j \ket{E_j}
 \label{E:finite-dim_Psi_form}
\end{equation}
 where $\{ E_j \}$ is a strictly increasing sequence of real numbers and
 $|a_j| \ne 0$ for all $j$.  We assume that all $a_j$'s and $E_j$'s are known.
 (Note that this assumption is not as restrictive as it appears.  For
 Hamiltonian $H$ and initial quantum state $\ket{\Psi(0)}$ express in another
 basis, $E_j$'s can be found, say, by Householder transformation and QR
 algorithm --- both are fast and numerically stable~\cite{Num}.  As for
 $a_j$'s, they are the projection of $\ket{\Psi(0)}$ on invariant subspaces of
 $H$.  Since finding invariant subspaces of a Hermitian matrix is efficient
 and numerically stable~\cite{CD79,BGK79}, so is computing $a_j$'s.)

 Note that $x^p$ is a strictly convex (concave) function of $x \ge 0$ if $1 <
 p \le 2$ ($0 < p < 1$).  And when $p = 1$, this function is both convex and
 concave.  Since the sum of convex (concave) functions is convex (concave),
 $\braket{|E-E_r|^p}$ is a piecewise convex (concave) function of $E_r$ for
 $1 \le p \le 2$ ($0 < p \le 1$).  Consequently, for $p > 1$,
 $\braket{|E-E_r|^p}$ attains its minimum at a certain $E_{r,\opt} \in
 (E_j,E_{j+1})$ for some $j = 1,2,\ldots,n-1$ where $E_{r,\opt}$ is the
 solution of the equation
\begin{equation}
 \eplus{(E-E_r)}{p-1} - \eplus{(E_r-E)}{p-1} = 0
 \label{E:E_r_opt_condition}
\end{equation}
 for $p \ge 1$.  It is easy to see that the L.H.S. of
 Eq.~\eqref{E:E_r_opt_condition} is strictly decreasing and continuous for $p
 > 1$.  Besides, $\eplus{(E_1-E)}{p-1} = 0 = \eplus{(E-E_n)}{p-1}$.  So,
 $E_{r,\opt}$ exists and is unique.  Furthermore, it can be evaluated
 accurately and efficiently via bisection method.  Nevertheless, as depicted
 in Fig.~\ref{F:E_r_curve}, the L.H.S. of Eq.~\eqref{E:E_r_opt_condition}
 makes sharp turns around and has infinite slopes at $E_j$'s.  (The latter can
 be proven by differentiating the L.H.S. of Eq.~\eqref{E:E_r_opt_condition}
 with respect to $E_r$.)  This greatly lowers the efficiency of bisection
 method if $E_{r,\opt}$ is near one of the $E_j$'s.  Here we recommend first
 using binary search to efficiently limit $E_{r,\opt}$ to one of the intervals
 $(E_j,E_{j+1})$.  Then, since the second derivative of the L.H.S. of
 Eq.~\eqref{E:E_r_opt_condition} with respect to $E_r$ changes sign exactly
 once in this interval, Lemma~\ref{Lem:Newton_method} implies that using
 either $E_j + \delta$ or $E_{j+1} - \delta$ for a sufficiently small $\delta
 > 0$ as initial guess, Newton's method is guaranteed to converge to
 $E_{r,\opt}$.  Indeed, this is what we have observed in our numerical
 experiments.

\begin{figure}[t]
 \centering\includegraphics[width=7.5cm]{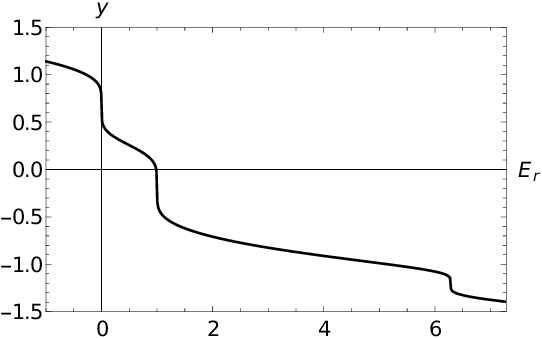}
 \caption{The L.H.S. of Eq.~\eqref{E:E_r_opt_condition} as a function of $E_r$
  for $p = 1.2$ and $\ket{\Psi(0)} = \sqrt{0.4}\ket{0} + \sqrt{0.45} \ket{1} +
  \sqrt{0.15} \ket{2\pi}$, namely, one of the initial states used in
  Sec.~\ref{Sec:performance} to study the performance of various \QSL{s}.  The
  shape of this curve is typical among finite-dimensional $\ket{\Psi(0)}$'s.
  Notice that the curve makes sharp turns with infinite slope at each $E_j$.
  \label{F:E_r_curve}}
\end{figure}

 As for the case of $p = 1$, the L.H.S. and R.H.S. of
 Eq.~\eqref{E:E_r_opt_condition} are only piecewise continuous in general.
 Thus, standard calculus technique does not work.  Fortunately,
 Eq.~\eqref{E:E_r_opt_condition} means that $E_{r,\opt}$ can be chosen to be
 the median energy of the state~\cite{Chau13}.  That is to say, $E_{r,\opt}$
 obeys $\sum_{j \colon E_j \ge E_{r,\opt}} |a_j|^2 \ge 1/2$ and $\sum_{j
 \colon E_j \le E_{r,\opt}} |a_j|^2 \ge 1/2$.  Note that in this case, the
 values of $\mu^\pm$ are unique even though $E_{r,\opt}$ need not be unique.

 The case of $0 < p < 1$ is computationally more involved.  Since
 $\braket{|E-E_r|^p}$ is a piecewise concave function, its minimum is attained
 at $E_r = E_j$ for some $j = 1,2,\ldots,n$.  That is,
\begin{equation}
 E_{r,\opt} \in \{ E_j \colon \braket{|E-E_j|^p} \le \braket{|E-E_\ell|^p}
 \ \forall \, \ell = 1,2,\ldots,n \} .
 \label{E:concave_E_r_opt_condition}
\end{equation}
 Finding $E_{r,\opt}$ is accurate whose time complexity scales as $\Oh(n^2)$.
 Unfortunately, there is no way to further reduce the computing time using a
 deterministic algorithm for $\braket{|E-E_j|^p}$ shows no trend in general.

 Another case of importance is when $H$ has a bounded and continuous energy
 spectrum, $\eplus{(E-E_r)}{p}$ and $\eplus{(E_r-E)}{p}$ are differentiable
 functions of $E_r$ with $p \in (0,2]$.  Moreover, $\eplus{(E-E_r)}{p-1}$ and
 $\eplus{(E_r-E)}{p-1}$ are continuous functions of $E_r$.  So the above four
 functions can be accurately and efficiently computed via any standard
 numerical integrator of choice.  Using the same argument in this Subsection
 plus differentiation under the integral sign, Eq.~\eqref{E:E_r_opt_condition}
 holds for any $p \in (0,2]$.  The only difference is that the L.H.S. and
 R.H.S. of Eq.~\eqref{E:E_r_opt_condition} are strictly decreasing
 (increasing) and strictly increasing (decreasing) functions of $E_r$ for $p <
 1$ ($p \ge 1$), respectively.  In both cases, Newton's method must converge
 to unique $E_{r,\opt}$ according to Lemma~\ref{Lem:Newton_method}.

\subsection{Computation Of $\boldsymbol{\varphi_p}$ And
 $\boldsymbol{\theta_\opt}$}
\label{Subsec:numeric_theta}
 We now study how to efficiently compute $\theta_\opt$ for fixed $\mu^\pm
 \equiv \mu^\pm(E_{r,\opt})$.  To do so, we need a way to calculate the
 optimal $\theta$ that maximizes the R.H.S. of
 Inequality~\eqref{E:CZ_bound_basic_generic_improved}.

\begin{theorem}
 Let $0 < p \le 1$.  Suppose $\braket{|E-E_r|^p} \ne 0$, and denote
 $\eplus{(E-E_r)}{p} : \eplus{(E_r-E)}{p} = \mu^+ : \mu^-$ with $\mu^+ + \mu^-
 = 1$.  Then
 \begin{equation}
  \max_{|\theta|\le \cos^{-1}(\sqrt{\epsilon})} \frac{\cos\theta -
  \sqrt{\epsilon}}{\mu^+ A_{p,\theta}^+ + \mu^- A_{p,\theta}^-}
   \label{E:max_theta_def}
 \end{equation}
 attains its maximum when $\theta$ is the unique solution of
 \begin{equation}
  s_\epsilon(\theta) \equiv \mu^+ s_\epsilon^+(\theta) + \mu^- 
  s_\epsilon^-(\theta) \equiv \frac{\mu^+ \sin\frac{\varphi_p^+ - \theta}{2}
  \left( \cos\frac{\varphi_p^+ - \theta}{2} - \sqrt{\epsilon}
  \cos\frac{\varphi_p^+ + \theta}{2} \right)}{(\varphi_p^+ - \theta)^p} +
  \frac{\mu^-\sin\frac{\varphi_p^- - \theta}{2} \left( \cos\frac{\varphi_p^- -
  \theta}{2} - \sqrt{\epsilon} \cos\frac{\varphi_p^- + \theta}{2}
  \right)}{(\theta - \varphi_p^-)^p} = 0
  \label{E:max_theta_alg}
 \end{equation}
 in the interval $[-\theta_\crit,\theta_\crit]$.  Here, $\theta_\crit$ is the
 unique solution of
 \begin{equation}
  r_\epsilon^-(\theta) \equiv \cos \left( \frac{\varphi_p^- - \theta}{2}
  \right) - \sqrt{\epsilon} \cos \left( \frac{\varphi_p^- + \theta}{2}
  \right) = 0
  \label{E:theta_crit_def}
 \end{equation}
 in the interval $[0,\cos^{-1}(\sqrt{\epsilon})]$.
 \label{Thrm:CZ_bound_computation}
\end{theorem}
\begin{proof}
 From Eqs.~\eqref{E:A_p_theta_expression}--\eqref{E:varphi_def}
 and~\eqref{E:dvarphi_dtheta}, we have
 \begin{equation}
  \frac{dA_{p,\theta}^+}{d\theta} = A_{p,\theta}^+ \left[
  \frac{d\varphi_p^+}{d\theta} \cot\varphi_p^+ - \frac{(p-1) \left(
  \frac{d\varphi_p^+}{d\theta} - 1 \right)}{\varphi_p^+ - \theta} \right] =
  \frac{p A_{p,\theta}^+}{\varphi_p^+ - \theta} \left( 1 -
  \frac{\sin\theta}{\sin\varphi_p^+} \right) = \frac{\sin\varphi_p^+ -
  \sin\theta}{(\varphi_p^+ - \theta)^p} .
  \label{E:dA_dtheta}
 \end{equation}
 Hence,
 \begin{subequations}
  \label{E:numerical_turning_point_expression}
 \begin{equation}
  A_{p,\theta}^+ \sin\theta + (\cos\theta - \sqrt{\epsilon})
  \frac{dA_{p,\theta}^+}{d\theta} = \frac{2\sin\frac{\varphi_p^+ - \theta}{2}
  \left( \cos\frac{\varphi_p^+ - \theta}{2} - \sqrt{\epsilon}
  \cos\frac{\varphi_p^+ + \theta}{2} \right)}{(\varphi_p^+ - \theta)^p} .
  \label{E:numerical_turning_point_expression_plus}
 \end{equation}
 Similarly, we have
 \begin{equation}
  A_{p,\theta}^- \sin\theta + (\cos\theta - \sqrt{\epsilon})
  \frac{dA_{p,\theta}^-}{d\theta} = \frac{2\sin\frac{\varphi_p^- - \theta}{2}
  \left( \cos\frac{\varphi_p^- - \theta}{2} - \sqrt{\epsilon}
  \cos\frac{\varphi_p^- + \theta}{2} \right)}{(\theta - \varphi_p^-)^p} .
  \label{E:numerical_turning_point_expression_minus}
 \end{equation}
 \end{subequations}
 Since the argument in Expression~\eqref{E:max_theta_def} is a differentiable
 function of $\theta \in
 (-\cos^{-1}[\sqrt{\epsilon}],\cos^{-1}[\sqrt{\epsilon}])$, by differentiating
 this expression with respect to $\theta$ with the help of
 Eqs.~\eqref{E:A_p_theta_expression}
 and~\eqref{E:numerical_turning_point_expression}, after some algebraic
 manipulation, we conclude that the extremum in
 Expression~\eqref{E:max_theta_def} occurs when $\theta$ obeys
 Eq.~\eqref{E:max_theta_alg}.

 Recall from Corollary~\ref{Cor:extended_cos_inequality} that
 $\varphi_p^-(\theta) = -\varphi_p^+(-\theta)$.  So from the proof of
 Theorem~\ref{Thrm:CZ_bound}, we know that $r_\epsilon^-(\theta)$ is a
 strictly decreasing function of $\theta \in [-\cos^{-1}(\sqrt{\epsilon}),
 \cos^{-1}(\sqrt{\epsilon})]$.  Applying compound angle formula twice and
 from Corollary~\ref{Cor:varphi_trend}, we obtain
 $\left. r_\epsilon^-(\theta) \right|_{\theta = \cos^{-1}(\sqrt{\epsilon})} =
 \cos[(\varphi_p^- - \theta)/2] - \cos\theta \cos[(\varphi_p^- + \theta)/2] =
 \sin[(\varphi_p^- + \theta)/2] \sin\theta \le 0$ and $\left.
 r_\epsilon^-(\theta) \right|_{\theta = 0} = (1-\sqrt{\epsilon})
 \cos(\varphi_p^-/2) \ge 0$.  Therefore, $\theta_\crit$ is the unique solution
 of Eq.~\eqref{E:theta_crit_def} in the interval
 $[0,\cos^{-1}(\sqrt{\epsilon})]$.  From Corollary~\ref{Cor:varphi_trend},
 $\sin [(\varphi_p^- - \theta)/2] \ge 0$ for $|\theta| \le \pi/2$.  In other
 words, $s_\epsilon^-(\theta) \ge 0$ for $\theta \in [-\pi/2,\theta_\crit]$
 and $s_\epsilon^-(\theta) < 0$ for $\theta \in (\theta_\crit,\pi/2]$.  As
 $\varphi_p^+(\theta) = -\varphi_p^-(-\theta)$, we know that
 $s_\epsilon^+(\theta) = -s_\epsilon^-(-\theta)$.  So $s_\epsilon^+(\theta)
 \ge 0$ for $\theta \in [-\pi/2,-\theta_\crit]$ and $s_\epsilon^+(\theta) < 0$
 for $\theta \in [\theta_\crit,\pi/2]$.  Hence, solutions of
 Eq.~\eqref{E:max_theta_alg}, if any, must lie in
 $[-\theta_\crit,\theta_\crit]$.

 Let us rewrite $s_\epsilon^+(\theta)$ in Eq.~\eqref{E:max_theta_alg} as
 $s_{\epsilon,1}^+(\theta) - \sqrt{\epsilon} s_{\epsilon,2}^+(\theta)$.
 Note that $s_\epsilon^-(\theta) = -s_\epsilon^+(-\theta)$.  Moreover, from
 the proof of Theorem~\ref{Thrm:CZ_bound_computation}, $s_\epsilon^+$ is
 differentiable for $\theta \in [-\cos^{-1}(\sqrt{\epsilon}),
 \cos^{-1}(\sqrt{\epsilon})]$.  Furthermore, $s_\epsilon^+(-\theta_\crit) \le
 0$ and $s_\epsilon^+(\theta_\crit) \ge 0$.  So Eq.~\eqref{E:max_theta_alg}
 has a unique solution in the domain $[-\theta_\crit,\theta_\crit] \subset
 [-\cos^{-1}(\sqrt{\epsilon}),\cos^{-1}(\sqrt{\epsilon})]$ if we could show
 that $d s_{\epsilon,1}/d\theta > 0$ and $d s_{\epsilon,2}/d\theta < 0$ for
 any $|\theta| < \theta_\crit$.

 Clearly, $\theta = \theta_\crit$ is the unique solution of
 Eq.~\eqref{E:max_theta_alg} in $[-\theta_\crit,\theta_\crit]$ in the case of
 $\theta_\crit = 0$.  So, we only need to consider the case of $\theta_\crit >
 0$.  Since $d s_{\epsilon,2}/d\theta < 0$ if
 \begin{align}
  \frac{d (\varphi_p^+ - \theta)}{d\theta} \cos \left( \frac{\varphi_p^+ +
   \theta}{2} \right) \left[ (\varphi_p^+ - \theta) \cos \left(
   \frac{\varphi_p^+ - \theta}{2} \right) - p \sin \left( \frac{\varphi_p^+ -
   \theta}{2} \right) \right] \nonumber \\
  - \frac{d (\varphi_p^+ + \theta)}{d\theta} ( \varphi_p^+ - \theta) \sin
   \left( \frac{\varphi_p^+ - \theta}{2} \right) \sin \left( \frac{\varphi_p^+
   + \theta}{2} \right) &< 0 .
  \label{E:ds2_dtheta_1}
 \end{align}
 From Corollary~\ref{Cor:varphi_trend}, we know that the second term in the
 L.H.S. of Inequality~\eqref{E:ds2_dtheta_1} is negative and $\cos
 [(\varphi_p^+ + \theta)/2] d(\varphi_p^+ - \theta)/d\theta < 0$.  Therefore,
 it remains to show that
 \begin{equation}
  \frac{\varphi_p^+ - \theta}{2} \cot \left( \frac{\varphi_p^+ - \theta}{2}
  \right) \ge \frac{p}{2}
  \label{E:ds2_dtheta_2}
 \end{equation}
 for all $|\theta| < \theta_\crit$.  Note that from
 Corollary~\ref{Cor:varphi_trend}, $(\varphi_p^+ - \theta)/2 \in [0, 3\pi/4)$
 is a decreasing function of $\theta$.  Therefore, according to
 Lemma~\ref{Lem:x_cotx_div_2}, Inequality~\eqref{E:ds2_dtheta_2} holds if this
 inequality is true for $\theta = \theta_\crit$.  From
 Eq.~\eqref{E:theta_crit_def}, Corollary~\ref{Cor:varphi_trend} and the fact
 that $\varphi_p^+(\theta) = -\varphi_p^-(-\theta)$, we see that $\varphi_p^+
 - \theta_\crit$ is maximized when $\epsilon = 1$.  This happens when
 $\theta_\crit = 0$.  From Eq.~\eqref{E:varphi_def}, $p (1 -
 \cos\varphi_{p,\crit}^+) = \varphi_{p,\crit}^+ \sin\varphi_{p,\crit}^+$.  In
 other words, $[\varphi_{p,\crit}^+ - \theta_\crit] \cot [(\varphi_{p,\crit}^+
 - \theta_\crit)/2] = p$.  Hence, Inequality~\eqref{E:ds2_dtheta_2} holds for
 all $|\theta| < \theta_\crit$.

 Surely, $d s_{\epsilon,1}/d\theta > 0$ provided that $(\varphi_p^+ - \theta)
 \cot (\varphi_p^+ - \theta) < p$ for all $|\theta| < \theta_\crit$.  By the
 similar argument in the previous paragraph, we know that $\varphi_p^+ -
 \theta$ is minimized if $\theta = -\theta_\crit$.  This happens when
 $\epsilon = 1$ and hence $\theta_\crit = 0$.  Therefore, $p (1 -
 \cos\varphi_{p,-\crit}^+) = \varphi_{p,-\crit}^+ \sin\varphi_{p,-\crit}^+$,
 where $\varphi_{p,-\crit}^+$ is the shorthand notation for
 $\varphi_p^+(-\theta_\crit)$.  Hence, $p = [\varphi_{p,-\crit}^+ +
 \theta_\crit] \cot [(\varphi_{p,-\crit}^++ \theta_\crit)/2 ] >
 (\varphi_{p,-\crit}^+ + \theta_\crit] \cot (\varphi_{p,-\crit}^+ +
 \theta_\crit)$.  This completes our proof.
\end{proof}

\begin{remark}
 From Eqs.~\eqref{E:prob_ampl_req} and~\eqref{E:max_theta_alg}, the $\theta$
 maximizing Expression~\eqref{E:max_theta_def}, which we denote by
 $\theta_\opt$, obeys
 \begin{equation}
  \mu^+ |a_-|^2 (\theta_\opt - \varphi_p^-)^p = \mu^- |a_+|^2 (\varphi_p^+ -
  \theta_\opt)^p
  \label{E:theta_opt_interpretation}
 \end{equation}
 In particular, for the initial quantum state in
 Eq.~\eqref{E:Psi_optimal_form}, Eq.~\eqref{E:theta_opt_interpretation}
 becomes $\eplus{(E-E_r)}{p} : \eplus{(E_r-E)}{p} = \mu^+ : \mu^-$.  In other
 words, from Theorem~\ref{Thrm:CZ_bound}, for any $p \in (0,1]$, $\epsilon \in
 [0,1]$ and for any ratio $\mu^+ : \mu^-$, the \CZ bound can be saturated.  As
 expected, $\theta_\opt$ is also equal to the optimized value of $\theta$ in
 Inequality~\eqref{E:CZ_bound_basic_generic} of the \CZ bound.
 \label{Rem:theta_opt_interpretation}
\end{remark}

\begin{remark}
 Note that Expression~\eqref{E:max_theta_def} attains its maximum when $\theta
 = \theta_\crit$ if $\mu^- = 1$ or $\theta = -\theta_\crit$ if $\mu^+ = 1$.
 In this regard, Theorem~\ref{Thrm:CZ_bound_computation} tells us that for a
 general $\mu^+$, $\theta_\opt$ lies between these two limiting cases.
 \label{Rem:mu_theta_relation}
\end{remark}

 To use Theorem~\ref{Thrm:CZ_bound_computation} to compute the \CZ bound
 efficiently, we first need to calculate $\varphi_p^\pm$.  Recall that for
 $(p,\theta)\in {\mathcal R} \setminus (2,0)$, $\varphi_p^+$ and $\varphi_p^-$
 are the unique roots of Eq.~\eqref{E:varphi_def} in the intervals
 $[|\theta|,\pi)$ and $(-\pi,-|\theta|]$, respectively.  In the former case,
 we use Newton's method with $\pi$ as the initial guess; and in the latter
 case, we use $-\pi$ as the initial guess instead.  This method is
 quadratically convergent.  Here we prove this claim for $\varphi_p^+$.  The
 case of $\varphi_p^-$ can be similarly proven.  Note that from
 Eqs.~\eqref{E:f_p} and~\eqref{E:f_pp} in the proof of
 Lemma~\ref{Lem:cos_inequality} in the Appendix, $f''_{p,\theta}(x) < 0$ for
 $(p,x) \in (0,2] \times [\pi/2,\pi]$ and $f'_{p,\theta}(\pi) \le 0$.  As
 $\varphi_p^+ \in [|\theta|,\pi)$ for $p \in (0,1]$,
 Lemma~\ref{Lem:Newton_method} implies that the root $\varphi_p^+$ of
 Eq.~\eqref{E:varphi_def} in the interval $[|\theta|,\pi)$ can always be found
 by Newton's method using the initial guess $\pi$ for the case of $p \in
 (0,1]$.

 For the case when the root $\varphi_p^+ \in [0,\pi/2)$, we know from the
 proof of Lemma~\ref{Lem:cos_inequality} that this can only happen when $p \in
 (1,2]$ and $\theta \in (-\pi/2,0]$.  In this case, Eqs.~\eqref{E:f_p}
 and~\eqref{E:f_pp} tell us that $f''_{p,\theta}(\pi/2) = \theta -
 \varphi_p(\theta) \le 0$ and $f'_{p,\theta}(\pi/2) = 1 - p < 0$.  From the
 proof of Lemma~\ref{Lem:cos_inequality} in the Appendix, we know that
 $f'_{p,\theta}(x) = 0$ has exactly one root, say, $x_1$ in $(0,\pi/2)$.  In
 addition, $f''_{p,\theta}(x) = 0$ implies $(x-\theta) \tan x = 2 - p$.
 Obviously, the L.H.S. of this equation is a bijection from $[0,\pi/2)$ to
 $[0,+\infty)$.  Together with the fact that $p \le 2$, we conclude that
 $f''_{p,\theta}(x) = 0$ has exactly one root in $(0,\pi/2)$.  Thus,
 $f''_{p,\theta}(x) \le 0$ for all $x \in [x_1,\pi/2]$.  Recall that
 $f''_{p,\theta}(x) \le 0$ for $x \in [\pi/2,\pi]$ and $f'_{p,\theta}(\pi) <
 0$.  Applying Lemma~\ref{Lem:Newton_method} to $f_{p,\theta}(x)$ in the
 interval $[x_1,\pi]$, we know that Newton's method converges using the
 initial guess $\pi$.

 Since $\theta$ and $\varphi_p^+$ are simple roots of
 Eq.~\eqref{E:varphi_def}, combined with Eq.~\eqref{E:dvarphi_dtheta},
 Newton's method is stable.  Moreover, the necessary conditions for the loss
 of significance, in this case due to ill-conditioning, are that $\varphi_p^+
 \approx \theta \approx \pi/2$ and $p \approx 1$.  To play safe, we may switch
 to bisection method in this situation.  Nonetheless, our numerical experiment
 shows that Newton's method is highly accurate in this case as well.  The same
 findings apply to the numerical computation of $\varphi_p^-$.  Nevertheless,
 this ill-conditioning issue does affect the computation of
 Expression~\eqref{E:max_theta_def}.  We shall discuss it in
 Sec.~\ref{Subsec:numeric_coef} below.

\begin{figure}[t]
 \centering\includegraphics[width=7.5cm]{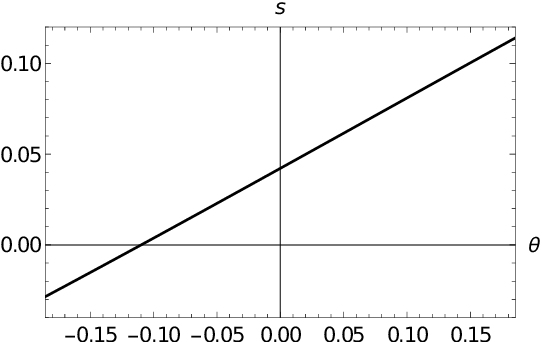}
 \caption{A typical $s(\theta)$ curve is very close to a straight line for all
  parameters used when $|\theta| \le \theta_\crit$.  We set $p = 0.7$,
  $\epsilon = 0.2$ and $\mu^+ = 0.8$ in this plot.
  \label{F:s_epsilon_curve}}
\end{figure}

 With $\varphi_p^\pm$ now accurately and efficiently evaluated, we can find
 $\theta_\opt$ by solving Eq.~\eqref{E:max_theta_alg}.  From the proof of
 Theorem~\ref{Thrm:CZ_bound_computation}, this can be done by bisection
 method using $[-\theta_\crit,\theta_\crit]$ as the initial interval.  We can
 do much better than this in practice.  As depicted in
 Fig.~\ref{F:s_epsilon_curve}, $s_\epsilon(\theta)$ looks like a straight line
 for $|\theta| \le \theta_\crit$.  Thus, Newton's method with an initial
 guess of, say, $\theta = 0$, will converge to $\theta_\opt$ quadratically.
 Another advantage of this method is that we do not need to numerically
 compute $\theta_\crit$.  In addition, for our parameter ranges, numerical
 stability is never an issue.  As for round-off error, the only situation that
 requires attention is when $|\varphi_p^\pm| \approx \theta$.  Here we may
 have to use the first few terms of the Taylor's series expansion of $\sin
 [(\varphi_p^\pm - \theta)/2]$ to accurately compute $\sin [(\varphi_p^\pm -
 \theta)/2]/|\varphi_p^\pm - \theta|^p$.

\subsection{Computation Of Expression~\eqref{E:max_theta_def}}
\label{Subsec:numeric_coef}
 Recall that the $\mu^\pm$ corresponding to the optimized $\braket{|E-E_r|^p}$
 may not be unique if $0 < p < 1$.  In this case, we need to maximize
 Expression~\eqref{E:max_theta_def} over all these $\mu^\pm$'s and hence the
 corresponding $\theta_\opt$'s.  Furthermore, irrespective of the uniqueness
 of $\mu^\pm$, from Eq.~\eqref{E:A_p_theta_def}, in order to evaluate
 Expression~\eqref{E:max_theta_def} accurately, highly accurate $\theta_\opt$
 and $\varphi_p^+(\theta_\opt)$ are required if $\varphi_p^+(\theta_\opt)
 \approx \pi$ or $\theta_\opt$ to compensate for the loss of significance in
 evaluating $A_{p,\theta_\opt}^\pm$ due to ill-conditioning and error
 propagation.  Similarly for $\varphi_p^-(\theta_\opt)$.  To compute
 Expression~\eqref{E:max_theta_def} to a certain accuracy in these situations,
 we may have to evaluate $\theta_\opt$ and $\varphi_p^\pm(\theta_\opt)$ to
 higher precision.  This can be done by changing the stopping criterion and
 perhaps also by increasing the working precision.  From
 Eq.~\eqref{E:varphi_def}, if $\varphi_p^+ \approx \pi$, then $p = \Oh(\pi -
 \varphi_p^+)$, which is small.  Conversely, if $p$ is small, then
 $\varphi_p^+ = \pi - \Oh(p)$.  That is to say, we need to evaluate
 $\varphi_p^+$ to a higher working precision to avoid round-off error if $p
 \approx 0$.  In fact, our numerical experiments suggest that using double
 precision arithmetic, rounding error is an issue, sometimes a serious one
 that gives totally wrong results, when $p \lesssim 10^{-5}$.  So we switch to
 quadruple precision arithmetic for $p < 10^{-5}$ in our Mathematica code.

 In summary, given $p, \epsilon$ and $\mu^+$, the method outlined in the
 previous paragraph can compute Expression~\eqref{E:max_theta_def} by
 numerically solving up to three equations, one for $\theta_\opt$ and another
 two for $\varphi_p^\pm(\theta_\opt)$.  (In case $\theta_\opt = 0$, one only
 needs to solve $\varphi_p^+$ as it is equal to $-\varphi_p^-$.)  And to
 evaluate the \CZ bound in Inequality~\eqref{E:CZ_bound_basic}, one has to
 further numerically find $E_{r,\opt}$ and hence $\mu^+$ by solving one more
 equation.  (In case $\mu^+$ is not unique, one has to solve $\theta_\opt$ and
 $\varphi_p^\pm$ for each $\mu^+$.)  Finally, one obtains the R.H.S. of
 Inequality~\eqref{E:CZ_bound_basic} by substitution.  In contrast, the most
 efficient way to obtain the \ML bound numerically is by solving just one
 single equation given $\epsilon$~\cite{Chau23}.  Can we adapt that method
 here?

 The reason why the \ML bound can be obtained by solving just one equation is
 that we can recast the problem as finding a normalized initial state that
 saturates the \ML bound.  Since any such state must belong to the Hilbert
 space spanned by two energy eigenstates of the Hamiltonian, only one degree
 of freedom remains, namely, the amplitude square of the ground state energy
 component of the normalized initial state.  More importantly, the exact
 evolution time can be written as an explicit function of this amplitude
 square as well as $\epsilon$.  Hence, computing the \ML bound is reduced to
 the problem of maximizing certain evolution time, which can further be
 reduced to the problem of solving a non-linear equation of one
 variable~\cite{Chau23}.

 Following the same logic, evaluating the \CZ bound in this way has to
 optimize over a normalized initial state in the form of
 Eq.~\eqref{E:Psi_optimal_form}, which has two degrees of freedom, namely,
 $|a_+|^2$ and $|a_-|^2$.  Unfortunately, the evolution time given $\epsilon$
 is an implicit function of $|a_\pm|^2$'s that we do not know how to write in
 an explicit form.  Thus, we can only proceed by solving a system of three
 equations, namely, the one relating the evolution time with $|a_\pm|^2$'s and
 two equations determined by maximizing the evolution time through varying
 $|a_\pm|^2$'s.  Consequently, there is no computational advantage over the
 method we have just presented.

\subsection{Computation Of The \CZ Bound For Fixed $\boldsymbol{p}$}
\label{Subsec:numeric_QSL_fixed_p}
 With the above discussion, it is clear that for a fixed $p$, the \CZ bound
 can be computed as follows.

 \begin{enumerate}
  \item Find $\braket{|E-E_{r,\opt}|^p}$ and $\mu^+$ either using the method
   in Sec.~\ref{Subsec:numeric_E_r} through determining $E_{r,\opt}$ or by an
   oracle that returns $\braket{|E-E_{r,\opt}|^p}$ and
   $\eplus{(E-E_{r,\opt})}{p}$.
  \item If $p \in (0,1]$, compute $\theta_\opt$ using the method in
   Sec.~\ref{Subsec:numeric_theta}.  Otherwise, set $\theta_\opt = 0$.
  \item Finally, compute Expression~\eqref{E:max_theta_def} and hence the \CZ
   bound using method in Sec.~\ref{Subsec:numeric_coef}.
 \end{enumerate}

 Evaluating the \LC bound can be done in almost the same way.  The only
 exception is that $\theta_\opt$ is always set to $0$.  For finite-dimensional
 quantum systems, this method is computationally accurate and efficient.

 Last but not least, we remark that our analysis here focuses on a fixed $p$.
 In spite of the fact that both $\lim_{p\to 0^+} \eplus{(E-E_r)}{p}$ and
 $\lim_{p\to 0^+} \eplus{(E_r-E)}{p}$ exist for any given $E_r$, the existence
 of finite \LC or \CZ bounds as $p\to 0^+$ is not guaranteed.  We shall
 report an interesting consequence of this observation in
 Sec.~\ref{Sec:performance}.

\subsection{Computation Of The Optimized \CZ Bound}
\label{Subsec:numeric_QSL_opt_p}
 Calculating the optimized \CZ bound is straightforward to implement but
 difficult to analyze.  Without additional information on $\ket{\Psi(0)}$, the
 powerful convex optimization method does not apply.  All we can do is to use
 a modern general optimization algorithm, such as differential evolution and
 Nelder-Mead method, over the parameter $p$ using the \CZ bound for fixed $p$
 as the target function.  Note that $\varphi_p^\pm$ vary smoothly with
 $\theta$.  If we further assume that $\braket{|E-E_r|^p}$,
 $\eplus{(E-E_r)}{p}$, $\eplus{(E-E_r)}{p-1}$ and $\eplus{(E_r-E)}{p-1}$ are
 continuous and monotonic functions of $p$, then we expect that any modern
 general optimization algorithm should work reasonably well both in speed and
 in accuracy.  One of us has posted the Mathematica code to compute the
 optimized \CZ bound~\cite{code}.  Based on this code, our numerical
 experiments to be reported in Sec.~\ref{Sec:performance} show that this is
 indeed the case although vigorous mathematical analysis is beyond reach for a
 general $\ket{\Psi(0)}$.  The same analysis applies also to the calculation
 of the optimized \LC and \LZ bounds.  Once again, our numerical experiments
 find that modern general optimization algorithms work very well for these two
 optimized bounds.

\subsection{Simple Expression For The Optimized \CZ Bound On Two-Dimensional
 Quantum Systems}
\label{Subsec:2D-CZ_expression}
 Interestingly, numerical analysis in this section gives us a simplified
 expression for the optimized \CZ bound on two-dimensional quantum systems.
 By shifting the reference level, we write the normalized initial state of
 such system as $a_1 \ket{E_1} + a_0 \ket{E_0}$ with $E_1 > E_0$.

 For the case of $p \le 1$, analysis in Sec.~\ref{Subsec:numeric_E_r} tells us
 that $E_{r,\opt}$ is either $E_0$ or $E_1$.  For the first subcase, $\mu^+$
 defined in Theorem~\ref{Thrm:CZ_bound_computation} has to be $1$; whereas for
 the second subcase, $\mu^+ = 0$.  So from
 Theorem~\ref{Thrm:CZ_bound_computation} and
 Remark~\ref{Rem:mu_theta_relation},
 \begin{equation}
  \left( \frac{\tau}{\hbar} \right)^p \ge \frac{1}{(E_1 - E_0)^p} \frac{\cos
  \theta_\crit - \sqrt{\epsilon}}{\min \left( A_{p,-\theta_\crit}^+ |a_1|^2,
  A_{p,\theta_\crit}^- |a_0|^2 \right)} ,
  \label{E:2d_p<1_case}
 \end{equation}
 with $\theta_\crit$ given by the unique solution of
 Eq.~\eqref{E:theta_crit_def}.

 If $1 \le p \le 2$, then $E_{r,\opt}$ obeys Eq.~\eqref{E:E_r_opt_condition}
 so that the \CZ bound becomes
 \begin{equation}
  \left( \frac{\tau}{\hbar} \right)^p \ge \frac{1 - \sqrt{\epsilon}}{(E_1 -
  E_0)^p A_{p,0}} \left( |a_0|^{-\frac{2}{p-1}} + |a_1|^{-\frac{2}{p-1}}
  \right)^{p-1} .
  \label{E:2d_p>1_case}
 \end{equation}
 Thus, the optimized \CZ bound for two-dimensional quantum systems can be
 expressed as
 \begin{equation}
  \frac{\tau}{\hbar} \ge \frac{1}{E_1 - E_0} \max \left\{ \max_{1\le p\le 2}
  \left[ \frac{(1 - \sqrt{\epsilon}) \left( |a_0|^{-\frac{2}{p-1}} +
  |a_1|^{-\frac{2}{p-1}} \right)}{A_{p,0}} \right]^{\frac{1}{p}} , \max_{0 < p
  \le 1} \left[ \frac{\cos\theta_\crit - \sqrt{\epsilon}}{\min \left(
  A_{p,-\theta_\crit}^+ |a_1|^2, A_{p,\theta_\crit}^- |a_0|^2 \right)}
  \right]^{\frac{1}{p}} \right\} .
  \label{E:CZ_bound_for_2d_case}
 \end{equation}
 A notable feature of this expression is that its R.H.S. is proportional to a
 factor that depends only on $\epsilon$ and $|a_0|$.

\begin{table*}[t]
 \caption{Various \QSL bounds for various initial states $\ket{\Psi(0)}$ and
  square root fidelities $\sqrt{\epsilon}$.  The initial states used in this
  Table from case~(a) to~(g) are $(\ket{0} + \ket{1})/\sqrt{2}$, 
  $\sum_{j=0}^{2047} \ket{j}/2^{11/2}$,
  $\sum_{j=1}^{2048} j^{-1} \ket{j} / \sqrt{\sum_{\ell=1}^{2048} \ell^{-2}}$,
  $\sqrt{0.1}\ket{0} + \sqrt{0.9}\ket{1}$,
  $\sqrt{0.3}\ket{0} + \sqrt{0.6}\ket{1} + \sqrt{0.1}\ket{2}$,
  $\sqrt{0.3}\ket{0} + \sqrt{0.6}\ket{1} + \sqrt{0.1}\ket{\pi}$ and
  $\sqrt{0.4}\ket{0} + \sqrt{0.45}\ket{1} + \sqrt{0.15}\ket{2\pi}$,
  respectively.
  Note that the \LZ bound is optimized over $p$, the \LC bound is optimized
  over $p$ and $E_r$, whereas the \CZ bound is optimized over $p$, $\theta$
  and $E_r$.  We do not tabulate the optimized values of $E_r$ here because
  relatively little can be learned from it.  Note that all \QSL bounds obtained
  are finite except for the optimized \LC and \CZ bounds in case number~(e)
  with $\sqrt{\epsilon} = 0.19$.  For the divergent case, the ``best
  estimation'' of the optimized \LC bound is $\approx 10^{4700}$ at $p
  \approx 10^{-6}$.  And the ``the estimation'' of the optimized \CZ bound is
  $\approx 10^{50000000}$ at $p \approx 10^{-10}$ and $\theta \approx
  10^{-11}$.
  \label{T:result}
 }
 \begin{longtable}{||l|c|D{.}{.}{2.8}|D{.}{.}{2.8}|D{.}{.}{2.8}|D{.}{.}{2.3}|D{.}{.}{2.8}|D{.}{.}{2.7}|D{.}{.}{2.8}|D{.}{.}{2.7}|D{.}{.}{2.7}||}
  \hline\hline
  case & $\sqrt{\epsilon}$ & 
   \multicolumn{1}{c|}{\MT} & \multicolumn{1}{c|}{\ML} &
   \multicolumn{2}{c|}{optimized \LZ} &
   \multicolumn{2}{c|}{optimized \LC} & \multicolumn{3}{c||}{optimized \CZ}
   \\
  \cline{5-11}
  & & \multicolumn{1}{c|}{bound} & \multicolumn{1}{c|}{bound} &
   \multicolumn{1}{c|}{bound} & \multicolumn{1}{c|}{$p_\opt$} &
   \multicolumn{1}{c|}{bound} & \multicolumn{1}{c|}{$p_\opt$} &
   \multicolumn{1}{c|}{bound} & \multicolumn{1}{c|}{$p_\opt$} &
   \multicolumn{1}{c||}{$\theta_\opt$} \\
  \hline
  (a) & 0 &
   3.1416 & 3.1416 &
   3.1416 & 1.78 &
   3.1416 & 1.2\times 10^{-5} &
   3.1416 & 8.8\times 10^{-4} & -5.6\times 10^{-4} \\
  \hline
  (b) & 0 &
   2.66\times 10^{-3} & 1.53\times 10^{-3} &
   1.88\times 10^{-3} & 2.00 &
   2.84\times 10^{-3} & 1.36 &
   2.84\times 10^{-3} & 1.36 & 0.00 \\
  \hline
  (c) & 0 &
   0.0353 & 0.2181 &
   0.4166 & 0.46 &
   0.3961 & 0.44 &
   0.4166 & 0.42 & -0.28 \\
  \hline
  (d) & 0.1 &
   4.9021 & 1.5432 &
   2.1437 & 2.00 &
   12.4204 & 1.00 &
   12.4204 & 1.00 & 0.00 \\
  \hline
  (e) & 0.19 &
   2.2994 & 1.5397 &
   1.8485 & 2.00 &
   \multicolumn{1}{c|}{$\infty^*$} & 0.00^* &
   \multicolumn{1}{c|}{$\infty^*$} & 0.00^* & 0.00^* \\
  \cline{2-11}
  & 0.20 &
   2.2824 & 1.5183 &
   1.8268 & 2.00 &
   3.1416 & 1.3\times 10^{-5} &
   3.1416 & 1.3\times 10^{-5} & 2.9\times 10^{-6} \\
  \hline
  (f) & 0.20 &
   1.5800 & 1.3287 &
   1.4586 & 1.75 &
   2.5970 & 1.2\times 10^{-5} &
   2.5970 & 1.0\times 10^{-5} & 2.2\times 10^{-6} \\
  \hline
  (g) & 0.00 &
   0.7461 & 1.1281 &
   1.1795 & 0.67 &
   1.3401 & 0.46 &
   1.3410 & 0.46 & 0.03 \\
  \cline{2-11}
  & 0.15 &
   0.6746 & 0.9342 &
   0.9323 & 0.89 &
   1.0211 & 0.73 &
   1.0221 & 0.74 & -0.04 \\
  \cline{2-11}
  & 0.35 &
   0.5762 & 0.6932 &
   0.6641 & 1.12 &
   0.7525 & 1.02 &
   0.7577 & 1.00 & -0.10 \\
  \cline{2-11}
  & 0.99 &
   0.0672 & 0.0099 &
   2.7\times 10^{-14} & 0.19 &
   0.0674 & 1.99 &
   0.0674 & 1.99 & 0.00 \\
  \hline\hline
 \end{longtable}
\end{table*}

\section{Performance Analysis}
\label{Sec:performance}
 Table~\ref{T:result} compares the minimum evolution times for various initial
 pure quantum states in energy representation using different \QSL{s}.
 (Recall from Sec.~\ref{Sec:prior_art} that the \dualML bound is basically the
 \ML bound of a time- and energy-reversed system.  Thus, we omit the \dualML
 bound in our table because such comparison is already reflected in the
 relative performance between \ML and \CZ bounds over various initial states.)
 These initial states are specifically chosen to illustrate our points.  From
 Table~\ref{T:result}, it is clear that the optimized \CZ bound is the best
 for all cases which is closely followed by the optimized \LC bound.  This
 result is consistent with our conclusion in Sec.~\ref{Sec:bound} that the
 optimized \CZ bound unifies all other bounds in the Table.  It also
 demonstrates that the \MT, \ML and optimized \LZ bounds are mutually
 complementary, and so are the optimized \LZ and optimized \LC bounds.

 Let us study Table~\ref{T:result} in detail.  The initial state of case
 number~(a) is $\ket{\Psi(0)} = (\ket{0} + \ket{1}) / \sqrt{2}$.  It saturates
 the \MT bound for any given of $\epsilon \in [0,1]$.  Table~\ref{T:result}
 shows that all the bounds we have covered give the saturation value of $\pi$
 when $\epsilon = 0$.  This is consistent with our conclusion in
 Theorem~\ref{Thrm:LC_bound_extension} that the optimized \LC and hence the
 optimized \CZ bounds are at least as good as the \MT bound.

 The initial state considered in case~(b) of Table~\ref{T:result} is
 $\sum_{j=0}^{2047} \ket{j} / 2^{11/2}$.  The time for it to its orthogonal
 complement equals $\tau = \pi\hbar/1024 \approx 0.0030680 \hbar$.  In this
 case, both the optimized \LC and \CZ bounds are the best.  They give about
 92.5\% of the actual evolution time whereas the \MT bound is just about
 86.6\%.  Here, both the optimized \LC and \CZ bounds give the same $p_\opt
 \approx 1.36 > 1$.

 Likewise, the (un-normalized) quantum state $\sum_{j=1}^{2048} j^{-1}
 \ket{j}$ used in case~(c) is also $2048$-dimensional.  (Here the dimension
 refers to the minimum Hilbert space dimension of the Hamiltonian needed to
 support such an initial quantum state.  We simply call this the Hilbert space
 dimension of the system and denote it by $n$.)  Its values of $p_\opt$ for
 both the optimized \LC and \CZ bounds are $\approx 0.89 < 1$.  We pick
 cases~(a)--(c) to test how the computational times vary as the Hilbert space
 dimension of the system $n$ increases for various \QSL{s} in practice.
 Moreover, we compare the practical efficiency of our methods for both $p_\opt
 < 1$ and $p_\opt \ge 1$ in case $n$ is large since the algorithms of finding
 $E_{r,\opt}$ introduced in Sec.~\ref{Subsec:numeric_E_r} and their
 corresponding complexities are vastly different in these two cases.  We use
 Mathematica code with just-in-time compilation installed in a typical desktop
 to compare their performance.  For the \MT and \ML bounds, increasing $n$
 from $2$ to $2048$ has relatively little effect on the computational times.
 The runtimes for the \ML bound are at most several ms in all three cases.  In
 contrast, the runtimes for the \MT bounds are about 20~$\mu$s, 0.9~ms and
 40~ms for cases~(a) to~(c), respectively.  For the unoptimized \LZ bound, the
 runtimes are similar to those of the \MT bound.  Whereas for the optimized
 \LZ bound, it increases from about 50~ms to about 0.3~s when $n$ increases
 from 2 to 2048.  As expected, longer times are needed to compute the bound
 for fixed $p$ and much longer to optimize $p$ as $n$ increases.  For example,
 for the unoptimized \LC bound, the runtimes are less than 1~ms for case~(a),
 about 70~ms for case~(b) and about 0.1~s for case~(c).  And for the optimized
 \LC bound, the runtimes are about 0.5~s, 150~s and 300~s for cases~(a)
 to~(c), respectively.  Last but not least, for the unoptimized \CZ bound, the
 increase is from about 3~ms for case~(a) to about 70~ms for case~(b) to about
 0.1~s for case~(c); while for the optimized \CZ bound, the increase is from
 about 5~s for case~(a) to 150~s for case~(b) to about 300~s for case~(c).  To
 conclude, our experiment shows that the unoptimized \LZ, \LC and \CZ bounds
 are all extremely efficient to evaluate.  And up to our expectation, longer
 time is required to compute both the optimized and unoptimized versions of
 the \LC and \CZ bounds when $p_\opt \le 1$ because of the higher
 computational complexity cost partly due to the existence of $n$ local minima
 in $\braket{|E-E_r|^p}$.  But in all cases, their optimized versions are fast
 enough to be used in the field even when the Hilbert space dimension of the
 system $n$ is of order of 1000.  Here we also mention on passing that simply
 using generic optimization method to obtain the optimized \LC and \CZ bounds
 is in general not practical when $n$ is large.  This is particularly true
 when $p_\opt \le 1$.  For example, computing the optimized \CZ bound in
 case~(c) using generic optimization takes about 1.3~hr even for probabilistic
 methods.  This is roughly 15~times longer than our algorithm.  Sometimes,
 generic optimization fails to produce an answer due to insufficient computer
 memory.  To be fair, using differential evolution, a generic probabilistic
 optimization technique, the runtime of the optimized \LC bound for case~(c)
 can sometimes be shortened to about 13~s.  We do not have a good explanation
 though.

 We now investigate how these \QSL{s} perform when it is not possible to
 evolve the given initial state to another state with fidelity $\epsilon$ in
 finite time.  Consider case~(d) with $\sqrt{\epsilon} = 0.1$ and case~(e)
 with $\sqrt{\epsilon} = 0.2$.  It is easy to show that the required evolution
 is not possible.  We pick these two cases to illustrate two points.  First,
 the optimized \LC and \CZ bounds can all diverge as $p\to 0^+$.  (We also
 see in some other cases not listed in Table~\ref{T:result} that the optimized
 \LZ bound diverges as well.)  Second, even for the case that these bounds
 cannot detect this impossible evolution, they generally give much higher
 \QSL{s}.  Next, we use case~(f) to test how various \QSL{s} handle another
 type of impossible evolution.  Specifically, we choose $\ket{\Psi(0)} =
 \sqrt{0.3} \ket{0} + \sqrt{0.6} \ket{1} + \sqrt{0.1} \ket{\pi}$ so that it
 could reach $\sqrt{0.3} \ket{0} - \sqrt{0.6} \ket{1} + \sqrt{0.1} \ket{\pi}$,
 namely, the only state whose fidelity $\epsilon$ is $0.2^2$ from it only in
 infinite time.  Table~\ref{T:result} shows that none of the bounds correctly
 give a divergent result even though the optimized \LC and \CZ bounds give
 identical \QSL that is significantly higher than the rest.  This is not
 surprising for given only information on the $p$th moments of the absolute
 value of energy of this type of systems, one may not have enough information
 to conclude that the evolution cannot be completed in finite time.  Note that
 the same conclusion can be drawn for case~(g) with $\epsilon = 0$ that we
 will cover in detail later in this Section.

 As a variation of the theme, we consider case~(e) with $\sqrt{\epsilon}$
 around $0.2$.  The evolution time $\tau$ needed is $\pi$.  (In fact, the
 evolution time is finite whenever $\sqrt{\epsilon} \ge 0.2$.)  Interestingly,
 Table~\ref{T:result} tells us that both the optimized \LC and \CZ bounds give
 this exact result.  More importantly, first order phase transition in the
 values of the optimized \LC and \CZ bounds are observed at $\sqrt{\epsilon} =
 0.2$.  This demonstrates the power of these two optimized bounds.

 The last case we consider is case~(g).  Here we fix $\ket{\Psi(0)} =
 \sqrt{0.4} \ket{0} + \sqrt{0.45} \ket{1} + \sqrt{0.15} \ket{2\pi}$ and vary
 $\epsilon$.  Surely, Table~\ref{T:result} shows that all bounds decrease as
 $\epsilon$ increases.  In addition, we find that the optimized \LC and \CZ
 bounds are the best.  They are generally better than the other \QSL{s} by
 about 10\%.  (And in some other cases listed in Table~\ref{T:result}, they
 can be about 30\% to 60\% better, sometimes even a few times better.)
 Besides, when $p_\opt \le 1$, the optimized \CZ bound is better than the
 optimized \LC bound by about 1\% because the optimized \CZ bound has the
 freedom to pick a non-zero $\theta_\opt$.  In addition, we see that $p_\opt$
 increases as $\epsilon$ increases for the optimized \LC and \CZ bounds.
 Table~\ref{T:result} also demonstrates that $p_\opt$ can take on any value in
 $(0,2]$, including the special case of $p_\opt = 1$.  These two trends are
 consistent with the fact that the evolution time $\tau$ decreases with
 increasing $\epsilon$ so that the optimized $p_\opt$ is likely to be the one
 with $\varphi_{p,\opt}^\pm$ both getting closer and closer to $\theta_\opt$.

\section{Conclusions And Outlook}
\label{Sec:conclusion}
 To summarize, we have proven the optimized \CZ bound that includes all
 existing \QSL{s} for time-independent Hamiltonian evolution as special cases.
 We also developed a precise and accurate numerical algorithm to compute this
 bound for quantum systems with underlying Hilbert space dimension $\lesssim
 2000$ and illustrated the usage of this bound through example initial states
 in Table~\ref{T:result}.  This optimized \CZ bound is at least as well as the
 existing ones and sometimes can be a few percent to a few times better.

 It is instructive to see how this bound can be used as a performance metric
 in realistic situation.  One possibility we have identified is in quantum
 control using piecewise constant pulse such as the one used in
 Ref.~\cite{SFWB12}.  This would be our follow up project.

\begin{acknowledgments}
 H.F.C. is supported by the RGC Grant No.~17303323 of the Hong Kong SAR
 Government.
\end{acknowledgments}

\medskip
\appendix
\section{Appendix}
\numberwithin{equation}{section}
\begin{proof}[Proof of Lemma~\ref{Lem:cos_inequality}]
 To show the validity of Inequality~\eqref{E:cos_inequality} for $x > \theta$,
 it suffices to prove that $A_{p,\theta}$ exists and equals
 \begin{equation}
  A_{p,\theta} = \sup_{x > \theta} m_{p,\theta}(x) \equiv \sup_{x > \theta}
  \frac{\cos\theta - \cos x}{(x-\theta)^p} .
  \label{E:A_p_theta_alt}
 \end{equation}
 Since $m_{p,\theta}(\pi) > 0$ for all $\theta \in (-\pi,\pi/2]$,
 $A_{p,\theta} > 0$ if it exists.  As $m_{p,\theta}(x) < 0$ for $\theta < 0$
 and $\theta < x < -\theta$, we only need to consider those $x \ge -\theta$
 for the supremum in Eq.~\eqref{E:A_p_theta_alt} if $\theta < 0$.  By fixing
 $b \in [-1,\cos\theta)$ for $\theta \in (-\pi,\pi/2]$, the set $S_{b,\theta}
 = \{ x > \theta \colon \cos x = b \}$ is non-empty and $\min S_{b,\theta} \in
 [|\theta|,\pi]$.  Observe that $m_{p,\theta}(x) > m_{p,\theta}(y)$ for all
 $x,y\in S_{b,\theta}$ with $x < y$ and $p > 0$.  So, $\lim_{x\to \theta^+}
 m_{p,\theta}(x)$ exists if $(p,\theta) \in {\mathcal R}$.  Consequently, by
 extending the definition of $m_{p,\theta}(x)$ to $x = \theta$ by continuity
 in the case of $\theta \ge 0$, the supremum in Eq.~\eqref{E:A_p_theta_alt} is
 in fact a maximum attained at a certain $\varphi_p(\theta) \in
 [|\theta|,\pi]$.  Since $dm_{p,\theta}(\pi)/dx < 0$, $A_{p,\theta} >
 m_{p,\theta}(\pi)$, we can safely omit $x = \pi$ in this maximization.  In
 this way, we obtain Inequality~\eqref{E:cos_inequality} for $(p,\theta) \in
 {\mathcal R}$.

 We rewrite Inequality~\eqref{E:cos_inequality} as
 \begin{equation}
  g_{p,\theta}(x) \equiv \cos x - \cos\theta + A_{p,\theta} (x-\theta)^p \ge 0
  \label{E:g_theta}
 \end{equation}
 for all $x \ge \theta$.  (Surely, Inequality~\eqref{E:g_theta} is trivially
 true for $x = \theta$.)  Suppose the maximum of the R.H.S. of
 Eq.~\eqref{E:A_p_theta_def} is reached when $x = \varphi_p(\theta)$.  For the
 time being, we do not assume that $\varphi_p(\theta)$ is unique.  We simply
 set $\varphi_p(\theta)$ to be any one of those $x$'s that maximizes the
 R.H.S. of Eq.~\eqref{E:A_p_theta_def}.  And we are going to prove its
 uniqueness in the next paragraph.  Clearly, $x = \theta$ and
 $\varphi_p(\theta)$ are zeros of the equation $g_{p,\theta}(x) = 0$ with the
 latter being a multiple root.  Hence, $g_{p,\theta}(\varphi_p) =
 g'_{p,\theta}(\varphi_p) = 0$.  This gives the expressions for $A_{p,\theta}$
 and $f_{p,\theta}$ in Eqs.~\eqref{E:A_p_theta_expression}
 and~\eqref{E:varphi_def}, respectively.

 We now prove the properties of the solutions of Eq.~\eqref{E:varphi_def} in
 Lemma~\ref{Lem:cos_inequality}.  Obviously, there is exactly one $x$
 maximizing Eq.~\eqref{E:A_p_theta_alt} if these properties are correct.  This
 is because more than one maximizing $x\in (|\theta|,\pi)$ means that there
 are at least two distinct $\varphi_p(\theta)$'s both in the same relevant
 domain satisfying Eq.~\eqref{E:varphi_def}.  This contradicts with the
 property that $\varphi_p(\theta)$ is unique.  We divide the remaining proof
 into the following five cases.  And we make use of the following equations,
 which are derived from Eq.~\eqref{E:varphi_def}.
 \begin{subequations}
 \begin{equation}
  f_{p,\theta}(0) = p (1 - \cos\theta) ,
  \label{E:f_at_0}
 \end{equation}
 \begin{equation}
  f_{p,\theta}(\theta) = 0 ,
  \label{E:f_at_theta}
 \end{equation}
 \begin{equation}
  f_{p,\theta}(-\theta) = 2\theta \sin\theta ,
  \label{E:f_at_minus_theta}
 \end{equation}
 \begin{equation}
  f_{p,\theta}(\frac{\pi}{2}) = \frac{\pi}{2} - \theta - p \cos\theta ,
  \label{E:f_at_pi_over_2}
 \end{equation}
 \begin{equation}
  f_{p,\theta}(\pi) = -p (1 + \cos\theta) ,
  \label{E:f_at_pi}
 \end{equation}
 \begin{equation}
  f'_{p,\theta}(\varphi_p) = (1-p) \sin\varphi_p + (\varphi_p - \theta)
  \cos\varphi_p
  \label{E:f_p}
 \end{equation}
 and
 \begin{equation}
  f''_{p,\theta}(\varphi_p) = (2-p) \cos\varphi_p - (\varphi_p - \theta)
  \sin\varphi_p .
  \label{E:f_pp}
 \end{equation}
 \end{subequations}

 \emph{Case~(a):} $(p,\theta) \in (0,1] \times (-\pi,0)$.  It is clear from
  Eq.~\eqref{E:f_pp} that $f''_{p,\theta}(\varphi_p)$ is a smooth function of
  $\varphi_p$ and $f''_{p,\theta}(\pi/2) \ne 0$.  Hence,
  $f''_{p,\theta}(\varphi_p) = 0$ implies that
  \begin{equation}
   (\varphi_p - \theta) \tan\varphi_p = 2 - p
   \label{E:f_pp_zero_condition}
  \end{equation}
  for $\varphi_p \in [0,\pi)$.  Obviously, the L.H.S. of
  Eq.~\eqref{E:f_pp_zero_condition} is a strictly increasing non-negative
  function in $[0,\pi/2)$ and a strictly increasing non-positive function in
  $(\pi/2,\pi]$, respectively.  As a result, Eq.~\eqref{E:f_pp_zero_condition}
  has a unique simple root $x_c$ in the interval $(0,\pi/2)$.  Combined with
  $f''_{p,\theta}(0) > 0$ and $f''_{p,\theta}(\pi) < 0$, we conclude that
  $f''_{p,\theta}(x) > 0$ for all $x \in [0,x_c)$ and $f''_{p,\theta}(x) < 0$
  for all $x \in (x_c,\pi)$.

  As $f_{p,\theta}(0), f'_{p,\theta}(0) > 0$, we deduce from the sign of
  $f''_{p,\theta}(x)$ for $x \in [0,x_c)$ that $f_{p,\theta}(x) > 0$ for all
  $x$ in this interval and $f'_{p,\theta}(x_c) > 0$.  According to
  Eq.~\eqref{E:f_at_pi_over_2}, $f_{p,\theta}(\pi/2)$ decreases as $\theta$
  increases from $-\pi$ to $0$.  Therefore, $f_{p,\theta}(\pi/2)
  \ge f_{p,0}(\pi/2) = \pi/2 - p > 0$.  From the sign of $f''_{p,\theta}(x)$
  for $x \in (x_c,\pi)$, we know that $f_{p,\theta}(x) > 0$ for all $x \in
  (x_c,\pi/2]$.  Besides, $f'_{p,\theta}(x)$ is strictly decreasing in
  $(\pi/2,\pi)$ according to the analysis in the last paragraph.  Together
  with $f_{p,\theta}(\pi) < 0$, we conclude that Eq.~\eqref{E:varphi_def} has
  a unique simple root in $(\pi/2,\pi)$.  This is because mean value theorem
  implies that Eq.~\eqref{E:varphi_def} has a root in $(\pi/2,\pi)$.  As the
  roots of Eq.~\eqref{E:varphi_def} in this interval forms a closed set, the
  smallest root exists which we denote by $x_r$.  Since $f_{p,\theta}(\pi/2) >
  0$ and $f'_{p,\theta}(x)$ is strictly decreasing in $(\pi/2,\pi)$, we know
  that $f'_{p,\theta}(x_r) < 0$.  Consequently, $x_r$ is the unique root of
  Eq.~\eqref{E:varphi_def} in $(\pi/2,\pi)$ because $f_{p,\theta}(x) > 0$ for
  all $x \in (\pi/2,x_r)$ and $f_{p,\theta}(x) < 0$ for all $x \in (x_r,\pi)$.

  Lastly, as $f_{p,\theta}(-\theta) > 0$, this unique solution of
  Eq.~\eqref{E:varphi_def} must lie in $[|\theta|,\pi)$ and hence in
  $[\max(|\theta|,\pi/2),\pi)$.  This completes the proof of the properties of
  roots of Eq.~\eqref{E:varphi_def} in this case.

 \emph{Case (b):} $(p,\theta) \in (0,1] \times (0,\pi/2]$.  Note that
  $f_{p,\theta}(\theta) = 0$ and $f'_{p,\theta}(\theta) \ge 0$ with equality
  holds if and only if $p = 1$.  Moreover, using the same argument to analyze
  $f''_{p,\theta}$ through Eq.~\eqref{E:f_pp_zero_condition} in the proof of
  case~(a), there is a $\bar{x}_c \in (0,\pi/2)$ such that
  $f''_{p,\theta}(\bar{x}_c) = 0$, $f''_{p,\theta}(x) > 0$ for $x \in
  (\theta,\bar{x}_c)$ and $f''_{p,\theta}(x) < 0$ for $x \in (\bar{x}_c,\pi)$.
  Therefore, for $f_{p,\theta}(\theta  + \delta), f'_{p,\theta}(\theta +
  \delta) > 0$ provided that $p \ne 1$ and $\delta > 0$ is sufficiently small.
  By Taylor's series expansion with remainder, the same is true for the case
  of $p = 1$.  Using similar argument in the proof of case~(a), we deduce that
  $f_{p,\theta}(\pi/2) \ge f_{p,\pi/2}(\pi/2) = 0$ with equality holds if and
  only if $\theta = \pi/2$.

  From the same argument using the sign of $f''_{p,\theta}$ in case~(a) and
  by using the fact that $f_{p,\theta}(\pi) < 0$, we conclude that
  $f_{p,\theta}(x) > 0$ for all $x \in (\theta,\pi/2)$.  Besides,
  Eq.~\eqref{E:varphi_def} has a unique simple root in $[\pi/2,\pi)$.

 \emph{Case~(c):} $(p,\theta) \in (1,2] \times (-\pi,0)$.  Using the same
  argument on $f''_{p,\theta}$ in the proof of case~(a), we know that there is
  a $\bar{\bar{x}}_c \in (0,\pi/2)$ such that $f''_{p,\theta}(\bar{\bar{x}}_c)
  = 0$, $f''_{p,\theta}(x) > 0$ for $x \in [0,\bar{\bar{x}}_c)$ and
  $f''_{p,\theta}(x) < 0$ for $x \in (\bar{\bar{x}}_c,\pi/2)$.  Note that
  $f_{p,\theta}(0), f_{p,\theta}(-\theta)$ and $f'_{p,\theta}(0) > 0$.  So
  using the same argument as in the proof of case~(a), we know that
  $f_{p,\theta}(x) > 0$ for all $x\in [0,|\theta|]$.  Together with the fact
  that $f_{p,\theta}(\pi) < 0$, we conclude that Eq.~\eqref{E:varphi_def} has
  a unique simple root in $[|\theta|,\pi)$.

 \emph{Case~(d):} $(p,\theta) \in (0,2) \times \{ 0 \}$.  In this case,
  $f_{p,0}(0) = f'_{p,0}(0) = 0$ and $f''_{p,0}(0) > 0$.  Hence,
  $f_{p,0}(\delta), f'_{p,0}(\delta) > 0$ for a sufficiently small $\delta >
  0$.  Together with $f_{p,0}(\pi) < 0$, we can use the same argument as in
  the proof of case~(a) to deduce the existence of a unique simple root for
  Eq.~\eqref{E:varphi_def} in $(0,\pi)$.  Furthermore, if $p \in (0,1]$, using
  the argument in the proof of case~(a), we know that $f_{p,0}(\pi/2) > 0$.
  Hence, the solution of Eq.~\eqref{E:varphi_def} can be further restricted to
  $[\pi/2,\pi)$.
 
 \emph{Case~(e):} $(p,\theta) = (2,0)$.  The argument is similar to that in
  the proof of case~(d).  Here, it is straightforward to check that
  $f^{(j)}_{p,0}(0) = 0$ for $j = 0,1,2,3$.  Besides, $f^{(4)}_{p,0}(0) < 0$.
  Therefore, $0$ is a root of Eq.~\eqref{E:varphi_def} of order~4.  From
  Eq.~\eqref{E:f_pp}, $f''_{p,0}(x) < 0$ for $x \in (0,\pi)$.  Therefore,
  $f'_{p,0}(x) < 0$ and $f_{p,0}(x) < 0$ for $x \in (0,\pi)$.  In other words,
  $0$ is the only root of Eq.~\eqref{E:varphi_def} in $[0,\pi)$.  This
  completes the proof of case~(e).
\end{proof}

\begin{proof}[Proof of Corollary~\ref{Cor:varphi_trend}]
 Here we only prove properties of those functions involving
 $\varphi_p^+(\theta)$.  As $\varphi_p^-(\theta) = -\varphi_p^+(-\theta)$,
 properties of those functions involving $\varphi_p^-(\theta)$ follows from
 the properties of the corresponding functions of $\varphi_p^+(\theta)$.
 Recall from Lemma~\ref{Lem:cos_inequality} that Eq.~\eqref{E:varphi_def} has
 a unique solution $\varphi_p^+(\theta) \in [-\pi/2,\pi)$ whenever $p \in
 (0,1]$ and $\theta \in [-\pi/2,\pi/2]$.  Therefore, $\varphi_p^+ \pm \theta
 \in [0,3\pi/2)$.  Moreover, the solution of Eq.~\eqref{E:varphi_def}, namely,
 $\varphi_p^+(\theta)$, is the one that maximizes the R.H.S. of
 Eq.~\eqref{E:A_p_theta_def}.  Thus, applying the implicit function theorem to
 Eq.~\eqref{E:varphi_def}, we know that $\varphi_p^+$ is differentiable and
 equals
 \begin{equation}
  \frac{d\varphi_p^+}{d\theta} = \frac{\sin\varphi_p^+ - p
  \sin\theta}{(\varphi_p^+ - \theta) \cos\varphi_p^+ + (1-p) \sin\varphi_p^+}
  \label{E:dvarphi_dtheta}
 \end{equation}
 provided that the denominator of this equation is non-zero.

 We claim that this denominator is non-positive for all $|\theta| \le \pi/2$
 and $p\in (0,1]$ with equality holds if and only if $p = 1$ and $\theta =
 \pi/2$.  In these parameter ranges, Lemma~\ref{Lem:cos_inequality} implies
 that $\varphi_p^+ \in [\pi/2,\pi)$ and $\varphi_p^+ \ge |\theta|$ with
 equality holds if $\theta = \pi/2$.  By multiplying the denominator of
 Eq.~\eqref{E:dvarphi_dtheta} by $\sin\varphi_p^+$ and by using
 Eq.~\eqref{E:varphi_def}, it suffices to prove that
 \begin{equation}
  \kappa(p,\theta) = \sin^2\varphi_p^+ - p + p \cos\theta \cos\varphi_p^+ \le
  0 .
  \label{E:dvarphi_dtheta_denom}
 \end{equation}
 (Note that through Eq.~\eqref{E:varphi_def}, we may regard $\kappa$ as a
 function of $p$ and $\theta$.)  From Eq.~\eqref{E:varphi_def}, we know that
 \begin{equation}
  \frac{\partial\theta}{\partial p} = \frac{\cos\varphi_p^+ -
 \cos\theta}{(1-p) \sin\theta} .
  \label{E:partialtheta_partialp}
 \end{equation}
 Hence,
 \begin{equation}
  \frac{d\kappa}{dp} = \cos\theta \cos\varphi_p^+ - 1 - \frac{p
  \cos\varphi_p^+ (\cos\varphi_p^+ - \cos\theta)}{1-p} < 0
 \end{equation}
 whenever $|\theta| < \pi/2$.  Therefore,
 \begin{equation}
  \kappa(p,\theta) < \lim_{q\to 0^+} \kappa(q,\theta) = \lim_{q\to 0^+} \sin^2
  \varphi_{q}(\theta) = 0
  \label{E:kappa_limit}
 \end{equation}
 for all $\theta \in [-\pi/2,\pi/2)$.  Here we obtain the last equality by
 solving Eq.~\eqref{E:varphi_def} in the limit of $q = 0^+$.  From
 Eq.~\eqref{E:varphi_def}, $[\varphi_p^+(\pi/2) - \pi/2] \tan
 \varphi_p^+(\pi/2) = -p$.  Using the analysis on the property of
 Eq.~\eqref{E:f_pp_zero_condition} in the proof of
 Lemma~\ref{Lem:cos_inequality}, we conclude that $\varphi_p^+(\pi/2) \ge
 \pi/2$ with equality holds if and only if $p = 1$.  So, $\kappa(p,\pi/2) \le
 \lim_{q\to 0^+} \kappa(q,\pi/2) = \sin^2\varphi_0^+(\pi/2) - p = 0$ with
 equality holds if and only if $p = 1$.  In summary, $\kappa(p,\theta) \le 0$
 for $p\in (0,1]$ and $|\theta| \le \pi/2$.  This proves our claim that the
 denominator of Eq.~\eqref{E:dvarphi_dtheta} is non-positive and equality
 holds if and only if $(p,\theta) = (1,\pi/2)$.

 As $\varphi_p^+$ is differentiable for $|\theta| \le \pi/2$, so to prove that
 $\varphi_p^+ - \theta$ is a strictly decreasing function of $\theta$ is
 equivalent to show that $\sin\varphi_p^+ - p \sin\theta > (\varphi_p^+ -
 \theta) \cos\varphi_p^+ - (1-p) \sin\varphi_p^+$.  Multiplying this
 inequality by $\sin\varphi_p^+$ and using Eq.~\eqref{E:varphi_def}, this is
 equivalent to proving that
 \begin{equation}
  2\sin^2\varphi_p^+ + p[\cos 2\varphi_p^+ - \cos(\varphi_p^+ - \theta)] > 0 .
  \label{E:varphi_trend1}
 \end{equation}
 Clearly, this inequality holds if $\cos 2\varphi_p^+ \ge \cos(\varphi_p^+ -
 \theta)$.  And for the case of $\cos 2\varphi_p^+ < \cos(\varphi_p^+ -
 \theta)$, it suffices to show that
 \begin{equation}
  0 < 2\sin^2\varphi_p^+ + \cos 2\varphi_p^+ - \cos(\varphi_p^+ - \theta)  = 1
  - \cos(\varphi_p^+ - \theta) .
  \label{E:varphi_trend2}
 \end{equation}
 Since $\varphi_p^+ > \theta$ unless $\theta = \pi/2$,
 Inequality~\eqref{E:varphi_trend2} is satisfied except possibly when
 $\varphi_p^+ = \theta = \pi/2$.  And in this case,
 Eq.~\eqref{E:dvarphi_dtheta} becomes $d\varphi_p^+/d\theta = 1$.  Therefore,
 $\varphi_p^+ - \theta$ is a strictly decreasing function of $\theta$.

 By the same token, we prove that $\varphi_p^+ + \theta$ is an increasing
 function of $\theta$ by showing that
 \begin{align}
  & \sin\varphi_p^+ - p \sin\theta + (\varphi_p^+ - \theta) \cos\varphi_p^+
   -(1-p) \sin\varphi_p^+ < 0 \nonumber \\
  \Longleftrightarrow{}& \cos(\varphi_p^+ + \theta) - \cos2\varphi_p < 0
   \nonumber \\
  \Longleftrightarrow{}& \sin\left( \frac{3\varphi_p^+ + \theta}{2} \right)
   \sin \left( \frac{\varphi_p^+ - \theta}{2} \right) < 0 .
  \label{E:varphi_trend3}
 \end{align}
  So, it suffices to prove that $3\varphi_p^+ + \theta \in (2\pi,4\pi)$.  Let
  us write $3\varphi_p^+ + \theta = 3(\varphi_p^+ - \theta) + 4\theta$.  Since
  $\varphi_p^+ - \theta$ is a strictly decreasing function of $\theta$ whose
  range is in $[0,3\pi/2)$, we conclude that $3\varphi_p^+ + \theta \in (2\pi,
  5\pi/2] \subset (2\pi,4\pi)$.
\end{proof}

\begin{proof}[Proof of Lemma~\ref{Lem:x_cotx_div_2}]
 Let $x \in (0,\pi)$.  Since $x > \sin x$, we have $2\cot (x/2) < x
 \csc^2(x/2)$.  This means $dh/dx < 0$ and hence $h$ is strictly decreasing in
 $[0,\pi]$.  Thus, $h \colon [0,\pi] \mapsto [0,2]$ is a homeomorphism.
\end{proof}

\begin{proof}[Proof of Lemma~\ref{Lem:opt_LC}]
 For the case of $\epsilon = 1$, the argument in the L.H.S. of
 Eq.~\eqref{E:opt_LC_aux} equals $0$ for all $x \in (0,\pi)$.  Thus,
 Eq.~\eqref{E:opt_LC_aux} holds in this case.

 For the case of $\epsilon \in [0,1)$, we consider the function
 \begin{equation}
  u(x) = \ln x + \frac{\ln \left( \frac{1-\sqrt{\epsilon}}{2} \right) - 2
  \ln\sin \frac{x}{2}}{x \cot \frac{x}{2}} ,
  \label{E:u_aux_def}
 \end{equation}
 which is smooth for $x \in (0,\pi)$.  Clearly, $u(x)$ is the logarithm of the
 argument in the L.H.S. of Eq.~\eqref{E:opt_LC_aux}.  In addition,
 \begin{equation}
  \frac{du}{dx} = \frac{1}{x^2} \left( \ln \frac{1-\sqrt{\epsilon}}{2} - 2
  \ln\sin \frac{x}{2} \right) \left[\frac{x \sec^2\frac{x}{2}}{2} - \tan
  \frac{x}{2} \right] .
  \label{E:dudx}
 \end{equation}
 Note that the first factor in the R.H.S. of Eq.~\eqref{E:dudx} is non-zero in
 $(0,\pi)$.  For the third factor to vanish, $x = \sin x$.  So, the third
 factor is non-zero in $(0,\pi)$, too.  As for the second factor in the R.H.S.
 of Eq.~\eqref{E:dudx}, it has a unique zero in $(0,\pi)$, namely, at $x =
 \cos^{-1}(\sqrt{\epsilon})$.  Since
 \begin{equation}
  \left. \frac{d^2 u}{dx^2} \right|_{x=\cos^{-1}(\sqrt{\epsilon})} =
  \frac{1}{[\cos^{-1}(\sqrt{\epsilon})]^2} \left[ 1-
  \frac{\cos^{-1}(\sqrt{\epsilon})}{\sqrt{1-\epsilon}} \right] < 0 ,
  \label{E:d2udx2}
 \end{equation}
 $u(x)$ and hence $\exp[u(x)]$ attain their global maxima in the interval
 $(0,\pi)$ at $x=\cos^{-1}(\sqrt{\epsilon})$.  As $\exp\{
 u(\cos^{-1}[\sqrt{\epsilon}]) \} = \cos^{-1}(\sqrt{\epsilon})$, we conclude
 that Eq.~\eqref{E:opt_LC_aux} is valid for $\epsilon \in [0,1)$.
\end{proof}

\begin{proof}[Proof of Lemma~\ref{Lem:Newton_method}]
 Denote the root by $x_r$.  Since this Lemma is trivially true when $x_r = b$,
 we may assume that $x_r < b$ and $w(b) < 0$.  By Taylor's series expansion
 with remainder, $0 = w(x_r) \le w(b) + w'(b) (x_r-b)$ or $x_r \le x_1 \equiv
 b - w(b)/w'(b)$.  We claim that $w(x_1) < 0$.  Suppose the contrary, $w(x) =
 0$ has another root $\bar{x}_r > x_1$.  Nonetheless, the above Taylor's
 series argument implies that $\bar{x}_r \le x_1$, which is absurd.  Note that
 $w'(x_1) < 0$.  If not, $w'(x) \ge 0$ for all $x \in [a,x_1)$ because $w''(x)
 \le 0$.  Then, by mean value theorem, $w(x) \le w(x_1) < 0$ contradicting the
 condition that $w(a) \ge 0$.  Replacing $b$ by $x_1$ and repeating the above
 argument, this Lemma can be proven by recursion.
\end{proof}

\bibliographystyle{apsrev4-2}
\bibliography{qsl2.3}

\end{document}